\documentclass[twocolumn,twoside]{IEEEtran}
\usepackage{amsmath,graphicx}
\usepackage{array}
\usepackage{framed}
\usepackage{amssymb}
\usepackage{amsthm}
\usepackage{mathtools,bm}
\usepackage{xcolor}
\usepackage{myTikz}
\usepackage{wrapfig}
\usepackage{pgfplotsthemetol}
\usepackage{pgfplotstable}
\usepackage{pgfplots}
\usepackage{csvsimple}
\usepackage{arydshln}
\usepackage{algorithm}
\usepackage{subcaption}
\usepackage{scalerel}
\usepackage{cite}
\newcommand{\response}[1]{\textcolor{black}{#1}}
\usepackage[noend]{algpseudocode}
\makeatletter
\def\BState{\State\hskip-\ALG@thistlm}
\makeatother
\floatname{algorithm}{Algorithm}

\usepackage{colortbl}
\usepgfplotslibrary{fillbetween}
\pgfplotsset{every tick label/.append style={font=\scriptsize}}
\pgfplotsset{every axis label/.append style={font=\scriptsize}}
\pgfplotsset{legend style={font=\scriptsize}}
\pgfplotsset{
	compat=1.11,
	legend image code/.code={
		\draw[mark repeat=2,mark phase=2]
		plot coordinates {
			(0cm,0cm)
			(0.25cm,0cm)        
			(0.5cm,0cm)         
		};%
	}
}
\DeclarePairedDelimiter\abs{\lvert}{\rvert}%
\DeclarePairedDelimiter\norm{\lVert}{\rVert}%

\theoremstyle{plain}

\newtheorem{lem}{Lemma}
\newtheorem{prop}{Proposition}

\newtheorem{rem}{Remark}
\newtheorem{defn}{Definition}

\theoremstyle{remark}

\title{Grid-Graph Signal Processing (Grid-GSP):\\
A Graph Signal Processing Framework\\ for the Power Grid}
\author{\IEEEauthorblockN{Raksha Ramakrishna} and \IEEEauthorblockN{Anna Scaglione}
	 \thanks{Work done at the School of Electrical, Computer and Energy Engineering (ECEE), Arizona State University, Tempe, AZ. This research was supported in part by the Director, Office of Electricity Delivery and Energy Reliability, 
	Cybersecurity for Energy Delivery Systems program, of the U.S. Department of Energy, under contract DOE0000780. Any
	opinions, and findings expressed in this material are those of the authors and
	do not necessarily reflect those of the sponsors
	Preliminary work was presented in  \cite{DSW2019,GlobalSIP2019}.}	\vspace{-0.7cm}}
\begin{document}
	%
	\maketitle
	\begin{abstract}
   The underlying theme of this paper is to explore the various facets of power systems data through the lens of graph signal processing (GSP), laying down the foundations of the Grid-GSP framework.  Grid-GSP provides an interpretation for the spatio-temporal properties of voltage phasor measurements, by showing how the well-known power systems modeling supports a generative low-pass graph filter model for the state variables, namely the voltage phasors. Using the model we formalize the empirical observation that voltage phasor measurement data lie in a low-dimensional subspace and tie their spatio-temporal structure to generator voltage dynamics.
   The Grid-GSP generative model is then successfully employed to investigate the problems pertaining to the grid of data sampling and interpolation, network inference, detection of anomalies and data compression. 
   Numerical results on a large synthetic grid that mimics the real-grid of the state of Texas, ACTIVSg2000, and on real-world measurements from ISO-New England verify the efficacy of applying Grid-GSP methods to electric grid data. 
	\end{abstract}
	\section{Introduction} \label{sec:introduction}	
	The power grid is one of the foremost examples of a large-scale man-made network. The nodes of the associated graph are the grid \textit{buses} and its edges are its \textit{transmission lines}. It is therefore natural to see measurements from the power grid as graph signals \cite{gridGSP_ortega} and model power grid measurements using tools from the theory of graph signal processing (GSP) whose goal is to extend fundamental insights that come from the frequency analysis for time series to the domain of signals indexed by graphs \cite{gridGSP_ortega,GFDesign,GSP_Moura}. 
	One of the factors that motivate the development of GSP for the power grid is the  abundance of high-quality data that can be acquired using  phasor measurement units (PMU)s, the sensors producing estimates of the voltage and current phasors  \cite{phadke2006history}. 
	With that, classical signal processing questions pertaining to sampling, interpolation, denoising and compression and questions that hinge on the underlying structure of the voltage phasors graph signal arise. 
 \par The overarching goal of this paper is to develop GSP based models for power systems from first principles by building upon the existing system-level knowledge of power systems to create a solid foundation to analyze power-grid measurements using tools from GSP. We call this the Grid-GSP framework.
 By identifying the correct graph shift operators (GSO), we  extend well-known results in GSP to power system data without losing the associated  physical interpretation. 
 
 \par The core idea is to rewrite the differential algebraic equations (DAE) \cite{GraphPowerSystem}, in a way often done in transient stability analysis of power systems,  to reveal that the inherent structure in voltage phasors can be explained using a linear low-pass graph filter as a generative model, \response{whose inputs are the generator voltages. This input signal is the generators' response to electric load in the grid.
 Through this model the paper shows also that the temporal dynamics of the input signal, i.e. the generator voltages, can be explained using a non-linear GSP model defined via another GSO derived from the  generator-only Kron-reduced network. This is done utilizing the well-known  classical swing equations \cite{Dorfler2013,kron_reduction_power_network}.  }
 \response{This spatio-temporal generative model supports the empirical observation that voltage data obtained using PMUs tend to be confined to a much smaller dimension compared to the size of the data record in both space and time \cite{Xie2014,LowrankMissingData}}. 
  Many papers have leveraged the empirical observation of the low-rank of phasor data for the interpolation of missing data  \cite{LowrankMissingData}, correcting bad data \cite{Wang2017} and to detect faulty events \cite{AnomalyDetection,RealTimeIdentification,Xie2014,Kim2015}.  Importantly, our framework explicitly puts forth the structure of this low-dimensional subspace using our GSP-based generative model, directly tying this subspace to the graph Fourier domain of the GSO. 

\subsection{Literature review} 
We review prior works by dividing the most relevant literature into three categories: 1) a general survey of works that use concepts from graph theory and GSP in power systems in the areas of sensor placement, interpolation and network inference, 2) FDI  attack detection and 3) literature pertaining to data compression for PMU data. 
\par \underline{\it Graph theory for power systems}: Several papers have used insights from spectral and algebraic graph theory. A few applications include optimal placement \cite{PMUplacement2010,Pal2014517} and generating statistically accurate topologies \cite{Wang2010}. 
Grid topology identification is a network inference problem and has been studied by several works such as in \cite{li2013blind,Grotas2019,Xiang2019,Deka2020,talukdar2020physics}.
 GSP concepts have been leveraged in \cite{GlobalSIP2018,Drayer2019} to detect FDI attacks.
Prior work in \cite{JSAC} dealt with performance limits on fault localization with  inadequate number of PMUs and connected it with graph signal sampling theory and optimal placement of PMUs for best possible resolution of fault localization in this under-sampled regime. 
\par The Kron-reduced network among the generator buses and the associated properties are used in \cite{InterAreaOsc,LocalOscHuang} to detect low-frequency oscillations as well as the resulting islanding patterns. In \cite{GraphLaplacian}, the authors have shed light on the relationship that exists between graph Laplacian and modes in power systems. 
Recently, a comprehensive review of graph-theoretical concepts in power systems is presented in \cite{GraphPowerSystem}. 
\par Additionally, there have been several papers adopting graphical models for state estimation \cite{weng2013graphical}, topology estimation \cite{deka2019topology} and optimal power flow \cite{dvijotham2017graphical}. While the modeling approach is valid, the graphical models capture correlation whereas our method models the underlying cause for that correlation structure thereby opening the door for statistical and non-statistical approaches. 

\par Note that with the exception of \cite{GlobalSIP2018,Drayer2019}, no other papers make the connection with GSP and even in the aforementioned work, GSP is used in an empirical manner. On the other hand, in our preliminary work in \cite{DSW2019,GlobalSIP2019}, established a case for GSP in a more rigorous manner.\\

\par \underline{\it Detection of FDI attacks}:
While anomaly detection can be broadly applied to identify various events, a large body of prior research has focused on False Data Injection (FDI) attacks that can bypass classical bad data detection (BDD) mechanisms \cite{bobba2010detecting,dan2010stealth,kosut2011malicious,liang2017review}  and trigger incorrect decisions or hide line overflows and contingencies \cite{Esnaola2016,Zhang2017,he2017real}). 
FDI attacks to PMUs can, for instance, exploit the vulnerability of GPS signals to spoofing attacks \cite{shepard2012evaluation,heng2014reliable}. \\
Recent papers on PMU data integrity have proposed leveraging the low rank spatio-temporal nature of PMU data not only to help with erasures but also to strengthen conventional BDD mechanisms \cite{Gao2016_PMUIdentification}. In \cite{Zhang2017} the authors have suggested an FDI attack strategy that can pass the aforementioned BDD approach in \cite{Gao2016_PMUIdentification} by generating false samples that approximately  preserve the original  subspace structure. \\

\par \underline{\it Compression of PMU data}: Due to their relatively high sampling rate and their wide deployment, PMU data have have called for compression. The performance of several off-the-shelf lossless encoding techniques applied to PMU data were investigated in \cite{Top2013}. 
In \cite{Klump2010}, a lossless compression called slack referenced encoding (SRE) method of PMU voltage measurements is introduced, by identifying a slack-bus and differentially encoding the difference between slack-bus measurements in time and all other buses. Similarly, in \cite{Tate2016}, phasor angle data is encoded in a lossless manner by preprocessing data using techniques from \cite{Klump2010} and then using Golomb-Rice entropy encoding.\\
The idea of slow-variation with time in PMU data  such as phase angles is used in \cite{KirtiHICCS} to firstly transform data into frequency domain (in time) and then using a `reverse water-filling' technique to encode frequency components in the difference measurements. 
Many lossy compression techniques utilize the low-rank structure inherent in voltage phasor data \cite{Gadde2016,Ge2015,DataCompressionSVD}. Several other wavelet based compression algorithms exist in the literature as well \cite{Mehra2013}. 
\subsection{Contributions}
The aim of this paper is to establish the framework of Grid-GSP and elucidate properties of power grid signals using tools from GSP. 
In particular, we: 
\begin{enumerate}
\item Establish that PMU voltage measurements from the power grid are result of an excitation to a low-pass graph filter whose graph shift operator (GSO) is defined using a function of the system admittance matrix.  
This was partly explored in our previous work in \cite{DSW2019,GlobalSIP2019} and used for blind community detection. 
\item Study the spatio-temporal structure of the excitation that, at a fast time scale, is dominated by the generators dynamics. It is shown that this excitation can be modeled as an auto-regressive graph filter \cite{isufi2017autoregressive} (GF-AR (2)) for the input signals from the generator internal bus.
The spatial properties in quasi-steady state are captured by  defining another  GSO using the `generator-only' or Kron-reduced \cite{kron_reduction_power_network} network for the generator buses.
\end{enumerate}
These models set the foundations to revisit known GSP based algorithms for sampling and reconstruction, interpolation and denoising and network inference in the context of signal-processing PMU data.
We harness the GFT for  feature extraction, to detect anomalies (specifically, FDI attacks \cite{GlobalSIP2019}) and to derive a lossy PMU voltage data compression algorithm, leveraging the sparsity of the GFT signal. 
By elucidating all the steps in modeling power systems data from graph construction, signal model, identification of the low-pass structure that is responsible for low-dimensional representation to signal denoising, network inference and anomaly detection, we illustrate how one can similarly develop  GSP based models in other application domains especially data that can be modeled as the output of \textit{low-pass} graph signals \cite{Ramakrishna2020}. 		
\vspace{-0.25em}
\subsection{Paper organization}
\vspace{-0.25em}
 Section \ref{sec:prelims} reviews concepts from GSP focusing on complex-valued graph signals and applicable more broadly to bandpass signals whose signal models rely on complex envelopes or phasors. It also reviews measurements and parameters pertaining to the grid.  
Section \ref{sec:gridGSP}  lays the foundation for Grid-GSP mapping the physical laws to a spatio-temporal generative model for voltage signals. Through these lens, in Section \ref{sec:GSP_algos_PMU_data}, the paper revisits algorithms and tools from GSP for PMU data pertaining to sampling and reconstruction along with optimal placement of PMUs, interpolation of missing samples and network inference. 
Section \ref{sec:appl_Grid_GSP} highlights applications of Grid-GSP to detect FDI attacks and for sequential lossy voltage data compression. 
The algorithms and methods are tested numerically in Section \ref{sec:results}. Section \ref{sec:conclusions} summarizes conclusions and  future research directions.
\par {\bf Notation}:  Boldfaced lowercase letters are used for vectors, $\bm{x}$ and uppercase for matrices, $\mathbf{A}$. Transpose is $\bm{x}^{\top},  \mathbf{A}^{\top}$ and conjugate transpose is $\mathbf{A}^{H}$.
$\left[\bm{x}\right]_{\mathcal{M}}$ is the new vector that has elements of the vector indexed by the set $\mathcal{M}$. The operation $\Re\{ . \} , \Im\{. \}$ denote the real and imaginary parts of the argument. 
 Pseudo-inverse of a matrix is  $\mathbf{A}^{\dagger}$. The operation $\texttt{diag}(\bm{x} )$ creates a diagonal matrix with elements from a vector $\bm{x}$ and $\texttt{Diag}(\mathbf{A})$ is the  extraction of diagonal values of a matrix $\mathbf{A}$.
\vspace{-0.435em}
\section{Preliminaries}\label{sec:prelims}
\subsection{Graph Signal Processing (GSP) in a nutshell}
Consider an \textit{undirected} graph $\mathcal{G} = \left( \mathcal{N}, \mathcal{E} \right)$ with nodes $i \in \mathcal{N}$ and edges $(i,j) \in \mathcal{E}$.  A graph signal $\bm{x} \in \mathbb{C}^{\abs{\mathcal{N}}}$ is a vector whose $i$th entry $\left[\bm{x}\right]_{i}$ is associated to node $i \in \mathcal{N}$. The set of nodes connected to node $i$ is called the \textit{neighborhood} of $i$ and denoted as $\mathcal{N}_{i}$. 
GSP generalizes the notion of {\it discrete time shift} for a time series by introducing the notion of graph shift operator (GSO):
\begin{defn}
 A graph shift operator (GSO) is  a linear neighborhood operator, so that each entry of the {\it shifted} graph signal is a linear combination of the graph signal neighbors' values \cite{GSP_Moura}. 
\end{defn}
The linear combinations can use complex-valued weights $s_{ij} \in \mathbb{C}, ~ \text{with}~ s_{ij} = 0, (i,j) \notin \mathcal{E}$ and, clearly, the GSO can be defined as matrix multiplication, with a matrix $\mathbf{S} \in \mathbb{C}^{\abs{\mathcal{N}} \times \abs{\mathcal{N}} }$.
Although not the only option,  one common choice for the GSO, is that of the graph weighted Laplacian\footnote{In truth, the Laplacian should be considered a graph differential operator as opposed to a shift operator, but we use the conventional name nonetheless in the rest of the paper to be consistent with the literature.}, i.e: 	 
 \begin{equation}
	\left[\mathbf{S}\right]_{i,j} = \begin{cases}
						\sum_{k \in \mathcal{N}_{i}} s_{i,k}, ~ i = j \\
						-s_{i,j}, ~ i \neq j
					\end{cases}
\end{equation}
 In this work, we focus on complex symmetric GSOs, $\mathbf{S} = \mathbf{S}^{\top}$ as is applicable to the weighted graph Laplacian for the power grid. Having defined the notion of shift, one can introduce the notion of \response{shift-invariance:}
 \begin{defn}
  Given a GSO $\mathbf{S}$ a {\it shift invariant operator} ${\cal H}$ acting on a graph signal is such that:
  \vspace{-0.8em}
\begin{equation}
{\cal H}: \bm{x} \mapsto \bm{v},~~~~\Longleftrightarrow~~~~{\cal H}: \mathbf{S}\bm{x} \mapsto \mathbf{S}\bm{v}. \label{eq:shift_invariance}
 \end{equation}
 \end{defn}
 \response{Linear shift-invariant operators must be  matrix polynomials of the GSO $\mathbf{S}$ \cite{sandryhaila2013discrete}. } Therefore, a linear shift-invariant graph filter is a linear operator and can be defined as:
\begin{equation}
\textstyle 	{\bm v}={\cal H}(\mathbf{S})\bm{x},~~	
		{\cal H}\left(\mathbf{S}\right)\!=\! \textstyle \sum\limits_{k=0}^{K-1}h_{k} \mathbf{S}^{k} \label{eq:graph_filter} 
\end{equation}	
Additionally, linear shift-invariant graph filters satisfy the condition: $	\mathbf{S}{\cal H}\left(\mathbf{S}\right)  = 	{\cal H}\left(\mathbf{S}\right) \mathbf{S}$.
\par Consider the following eigenvalue decomposition of the complex symmetric GSO $\mathbf{S}$, given by Theorem $4.4.13$ in \cite{Matrix_analysis} for  diagonalizable complex symmetric matrices:
\begin{align}
\mathbf{S} &= \mathbf{U} \bm{\Lambda} \mathbf{U}^{\top}, &	\mathbf{U}^{\top}\mathbf{U} = \mathbf{U}\mathbf{U}^{\top} = \mathbb{I}. \label{eq: Lap_eig}
\end{align}
Here $\bm{\Lambda}$ is the diagonal matrix with eigenvalues $\lambda_{0}, \lambda_{1}, \dots \lambda_{\abs{\mathcal{N}}-1 }$ on the principal diagonal  and  $\mathbf{U}$ are \textit{complex orthogonal} eigenvectors.
An equivalent concept of frequency domain in GSP is defined using eigenvalues and eigenvectors of the GSO. 
\par Graph frequencies are the eigenvalues of the GSO and the  order of frequencies is based  on the total variation (TV) criterion \cite{GSP_Moura,mallat2008wavelet} defined using the discrete $p$ Dirichlet form $S_{p}\!\!\left(\bm{x}\right)$ with $p=1$ as in \cite{Directed_graph_laplacian} as:
\begin{align}
S_{1}\!\!\left(\bm{x}\right)\!=\! \norm{\mathbf{S}\bm{x}}_{1} \implies S_{1}\!\left(\bm{u}_{i}\right) =  \norm{\mathbf{S}\bm{u}_{i}}_{1} = \lvert \lambda_{i} \rvert \norm{\bm{u}_{i}}_{1}
\end{align} 
After normalizing the eigenvectors such that $\norm{\bm{u}_{i}}_{1} = 1 ~\forall i$, it is clear that 
$S_{1}\left(\bm{u}_{i}\right) > S_{1}\left(\bm{u}_{j}\right) \implies \lvert \lambda_{i} \rvert > \lvert \lambda_{j} \rvert.$
Hence, the \textit{ascending} order of eigenvalues corresponds to increase in frequency,
$\lvert \lambda_{0} \rvert=0 \leq 	\lvert \lambda_{1} \rvert \leq 	\lvert \lambda_{2} \rvert \dots 	\lvert \lambda_{\abs{\mathcal{N}}} \rvert$.
This ordering  is not unique since two distinct complex eigenvalues can have the same magnitude.
\par The Graph Fourier Transform (GFT) basis is the complex orthogonal basis $\mathbf{U}$ in \eqref{eq: Lap_eig}. Hence, the  GFT of a graph signal $\bm{x}$, $\tilde{\bm{x}}$ and the inverse GFT are given by $\tilde{\bm{x}} = \mathbf{U}^{\top}\bm{x}$ and $\bm{x} = \mathbf{U} \tilde{\bm{x}}$ respectively where $\left[ \tilde{\bm{x}}\right]_{m} $ is the frequency component that corresponds to the $m$-th eigenvalue $\lambda_{m}$\footnote{It is worth noting that the graph shift operator and Fourier transforms do not have in general important properties that are found in their conventional counterparts for time series.  One notable fact is that the spectrum of ${\bm S}\bm x$ does not have the same amplitude as the spectrum of $\bm x$. In fact, the GSO effect is closer to that of a derivative, since each of the GFT coefficients is rescaled by the corresponding frequency. For complex symmetric, rather than Hermitian operators, unfortunately also Parseval theorem is not valid.}. Also, we can define the \textit{graph-frequency} response of the graph filter, $\tilde{\bm{h}}$, by writing
\begin{align}
\hspace{-2pt}\textstyle	\mathcal{H}\left(\mathbf{S}\right)&\!=\!\! \textstyle \sum\limits_{k=0}^{K} h_{k} \mathbf{S}^{k} \!=\! \mathbf{U} \texttt{diag}( \tilde{\bm{h}}) \mathbf{U}^{\top},\\ 
[\tilde{\bm{h}}]_{i} &\!\!=\! \!\mathrm{H}(\lambda_i),~~~~ \mathrm{H}(\lambda):=\sum\limits_{k=0}^{K}h_{k} \lambda^{k}  \label{eq:graph_filter_freq} 
\end{align}	
$ h_{k} \in \mathbb{C}, ~i =1,2,\dots \abs{\mathcal{N}}  $.
The frequency response of the filter is given by elements in $\tilde{\bm{h}}$. Subsequently, the input and output of a graph filter in graph-frequency domain are related as
\begin{align}
{\bm v} = {\cal H}(\mathbf{S})\bm{x} ~~\rightarrow ~~\tilde{ \bm{v}  } = \texttt{diag} (\tilde{\bm{h}}) \tilde{\bm{x}}, 
\end{align}
which is analogous to convolution theorem for time-domain signals. Naturally, this leads to the extension of notions such as low-pass, high-pass and band-pass filters and signals that are at the \response{heart} of sampling and interpolation schemes. 

\subsection{GSP for time series of graph signals}\label{subsec:GSP_time}
\response{So far, only the nodal index for the graph signal  $\bm{v}$ was considered. However, one can also encounter graph signal processes i.e. temporal variations in a graph signal  $\left\{\bm{v}_t\right\}_{t\geq 0}$. Since we are interested in the temporal characterization of voltage graph signals, we revise GSP concepts that are applied to time series of graph signals \cite{isufibanelli2016,isufi2017filtering} in this subsection. Then, we utilize these concepts while modeling the temporal dynamics at generator buses in Section \ref{subsec:generation}.}

\response{In order to characterize graph signal process $\{ {\bm v}_t \}_{t \geq 0}$, a joint time-vertex domain  is considered in the literature by defining filters whose response is shift invariant with respect to the time series shift operator $z^{-1}$ and an appropriately chosen GSO \cite{JFT}. 
To study the same,  map the time series of graph signal $\bm{v}_t$ in both the \response{graph frequency (GF)} and $z-$domain  by the application of $z$-transform to the GFT of the graph signal process:  }
\begin{align}
\textstyle \bm{V}(z)=\sum_{t=0}^{+\infty}\bm{v}_tz^{-t}, \widetilde{\bm{V}}(z) = \mathbf{U}^T\bm{V}(z),
\end{align}	
\response{We refer to $\widetilde{\bm{V}}(z)$ as the $z$-GFT. }
\response{A graph temporal filter's \cite{isufi2016separable}  impulse response ${\cal H}_t(\mathbf{S})$ and output $\bm{v}_t$ are
\begin{align}
{\cal H}_t(\mathbf{S}) = \sum_{k=0}^{K} h_{k,t}\mathbf{S}^{k}, ~\bm{v}_t = \sum_{\tau = 0}^{t} {\cal H}_{t-\tau}(\mathbf{S}) \bm{x}_{\tau},
\end{align}
respectively. Graph filter output $\bm{v}_t$  in the $z$-domain is:}
\begin{align}
\!\!\!\!\bm{V}(z)\!=\!{\cal H}(\mathbf{S}\!\otimes\! z)\bm{X}(z), \text{where}~
{\cal H}(\mathbf{S}\!\otimes\! z)\!:=\!
\sum_{t=0}^{+\infty}\!{\cal H}_t(\mathbf{S})z^{-t} \label{eq:z_domain_MIMO}
\end{align}
when the input is $\bm{x}_t$ with $z$-transform ${\bm{X}}(z)$
\response{and} ${\cal H}_t(\mathbf{S})$ are matrix polynomials of the GSO operator:
\begin{equation}
  \textstyle   {\cal H}_t(\mathbf{S})=\sum_{k=0}^{K}
    h_{k,t}\mathbf{S}^{k}~~\leftrightarrow~~
    {\cal H}(\mathbf{S}\otimes z)=
    \sum\limits_{k=0}^{K}H_k(z)\mathbf{S}^{k}.
\end{equation}
Here $H_k(z)$ is the $z$-transform of the filter $h_{k,t}$. We can define also the following impulse response in the GF domain:
\begin{equation}
\textstyle [\tilde{\bm h}_t]_i=\mathrm{H}_t(\lambda_i),~~~~\mathrm{H}_t(\lambda):=\sum_{k=0}^{
K}h_{k,t}\lambda^k \label{eq:GF_h}
\end{equation}
and the graph-temporal joint transfer function in the $z$ and \response{GF domain} as:
\begin{equation}
\textstyle [\tilde{\bm h}(z)]_i=\mathbb{H}(\lambda_i,z),~~\mathbb{H}(\lambda,z)=
\sum_{t=0}^{+\infty}\sum_{k=0}^{K}h_{k,t}\lambda^k z^{-t}
\end{equation}
With that, we obtain following input-output relationship:
\begin{align}
    \tilde{\bm{V}}(z)=
    \texttt{diag}
    \left(
    \tilde{\bm h}(z)
    \right)
    \tilde{\bm{X}}(z),
\end{align}
 by applying GFT to  $z$-domain  in \eqref{eq:z_domain_MIMO}.
\par  \response{In this work, we focus on a  class of graph-temporal filters called GF-ARMA $(q,r)$ filter \cite{isufi2016separable,isufi2017autoregressive} }. The input-output relation in both time  and $z$-GFT domain are described below, respectively:
\begin{equation} \notag
\begin{split}
& {\bm v}_t\!-\! {\cal A}_1( \mathbf{S} ) {\bm v}_{t-1}\!\cdots\! -\! {\cal A}_q (\mathbf{S}) {\bm v}_{t-q} \!\! =\!\! {\cal B}_0( \mathbf{S} ) {\bm x}_t \!\!+ \cdots \!+\! {\cal B}_r( \mathbf{S} ) {\bm x}_{t\!-\!r}, \\
&
  \texttt{diag}
  \left(
  \tilde{\bm a}(z)
  \right)
  \tilde{\bm{V}}(z)=
  \texttt{diag}
  \left(
  \tilde{\bm b}(z)
  \right)
  \tilde{\bm{X}}(z),
\end{split}
\end{equation}
where $\widetilde{\bm a}(z)=1-\sum_{t=1}^q \widetilde{\bm a}_t z^{-t}$ and  $\widetilde{\bm b}(z)=\sum_{t=0}^r \widetilde{\bm b}_t z^{-t}$ are the $z$-transform of the graph frequency responses of the graph filter taps $\{{\cal A}_t (\mathbf{S}) \}_{t=1}^q$, $\{{\cal B}_t (\mathbf{S}) \}_{t=0}^r$ for the GF-ARMA $(q,r)$ filter. \response{Particularly, the GF-AR (2) filter is used in Section  \ref{subsec:generation} to describe generator temporal dynamics.}
\vspace{-1em}
\subsection{Measurements and parameters of the electric grid}\label{ref:voltPMU}
\response{The electric grid network  can be represented by an undirected graph $\mathcal{G} = (\mathcal{N},\mathcal{E})$ where nodes are \textit{buses} and its edges are its \textit{transmission lines}.
The vertex set is a union between set of generator, $\mathcal{N}_G$ and non-generator/ load buses, $\mathcal{N}_L$, $ i \in \left\{\mathcal{N}_G \cup \mathcal{N}_L \right\} = \mathcal{N}$ and the edge set $(i,j) \in \mathcal{E}$ depicts electrical connections. 
\response{To obtain  Ohm's law for a network of transmission lines, one starts from the telegrapher equations for a single line to obtain the so-called ABCD parameters that relate input-output currents and voltages in the Fourier domain. 
The equations are then rearranged and the so-called $\pi$-model is attained, which is an equivalent circuit containing a series impedance element and parallel susceptance elements.  The  $\pi$-model  leads to the $2\times 2$ branch admittance matrix  that relates current and voltage injections at the {\it from} and {\it to} ends of a transmission line \cite{GloverSarmaOverbye:Book}.}
 \response{From the branch admittance matrix of the network, applying Kirchhoff's law, one can relate the current and voltage phasors for the entire network, introducing a \textit{system admittance matrix}, $\bm{Y} \in \mathbb{C}^{\abs{\mathcal{N}}}$\cite{GloverSarmaOverbye:Book}  thus obtaining the network version of Ohm's law (see \eqref{eq:Ohm}). The matrix $\bm{Y} $ is defined as:
\begin{align}
\left[\bm{Y}\right]_{i,j} = \begin{cases}
\sum_{k \in \mathcal{N}_{i}} y_{i,k}, ~ i = j \\
-y_{i,j}, ~ i \neq j
\end{cases}
\end{align}
where $y_{i,j}$ is the admittance of the branch between buses $i$ and $j$ if  $(i,j) \in \mathcal{E}$.  }
The system admittance matrix $\bm{Y}$ is a complex symmetric matrix and it is equivalent to the complex-valued graph Laplacian matrix associated with the power grid. }
Next, we will partition the nodes or buses into generator and non-generators, so that:
\begin{align}
\bm{Y} = 
\begin{bmatrix}
			\bm{Y}_{gg} & \bm{Y}_{g\ell} \\
			\bm{Y}^{\top}_{g\ell}& \bm{Y}_{\ell \ell} 
			\end{bmatrix} , 
\end{align}
where $\bm{Y}_{gg}$ is the generator buses-only network, $\bm{Y}_{g\ell}$ includes the portion connecting generators and loads and $\bm{Y}_{\ell \ell}$ corresponds to the section of the grid connecting the loads buses among themselves. 
The shunt \response{(fixed admittance to ground at a bus)} elements at all generator buses  are denoted by $\bm{y}_{sh}^{g} \in \mathbb{C}^{\abs{\mathcal{N}_G}}$ and at all load buses by $\bm{y}_{sh}^{\ell} \in \mathbb{C}^{\abs{\mathcal{N}_L}}$.
\par The state of the system, from which all other physical quantities of interest can be derived, are the voltage phasors at each bus.
\response{In the following we assume that a \response{PMU}  installed on node/bus $i \in \mathcal{N}$ provides a noisy measurement of voltage and current phasors at time $t$ where $v(t,i) = \abs{v(t,i)} e^{j\theta(t,i)}$. With some abuse of notation, we will refer to the PMU data as $v(t,i)$ as well.   Let the vector of voltage phasors collected at time $t$ be $\bm{v}_{t} \in \mathbb{C}^{\abs{ \mathcal{N} }}$. 
After $\bm{v}_{t}$ is partitioned into voltages at generator and non-generator buses,  let $\bm{i}^{g}_{t} \in \mathbb{C}^{\abs{\mathcal{N}_G}}$ be the generator current and $\bm{i}^{\ell}_{t} \in \mathbb{C}^{\abs{\mathcal{N}_L}}$ the load current. }
Ohm's law for a network is\footnote
{\response{Note that the admittances values are frequency responses evaluated at $60$Hz (for the US) the voltage and current signals are the corresponding envelopes at the same frequency; hence the assumption is the voltage and currents are narrowband and the convolution can be approximated by gain and phase rotation equal to the  Fourier response at $60$ Hz.}}: 
\begin{align}\label{eq:Ohm}
\hspace{-1.1em}\left(\!\bm Y \!+\!
\texttt{diag}\! \left( \begin{bmatrix}
 \bm{y}_{sh}^{g} \\
\bm{y}_{sh}^{\ell}
\end{bmatrix}\right)\!\right)\bm{v}_{t} \! =\!  \bm{i}_{t}, \text{where}~ \bm{v}_{t} \!=\! \begin{bmatrix}
\bm{v}^{g}_{t}  \\ \bm{v}^{\ell}_{t} 
\end{bmatrix},\! \bm{i}_{t}\! =\! 
\begin{bmatrix}
\bm{i}^{g}_{t} \\ 
\bm{i}^{\ell}_{t}
\end{bmatrix}
\end{align}
To describe the operating conditions of the system we introduce a few more quantities.  In power systems transient dynamic analysis the impact of generating units is modeled  as an internal bus characterized by a generator impedance (or admittance)
$\bm{y}_{g} \in \mathbb{C}^{\abs{\mathcal{N}_G}}$ for $g \in \mathcal{N}_G$  
connected to an ideal voltage source called {\it internal voltage}\cite{GloverSarmaOverbye:Book}; we denote its value at time $t$ by \response{$E(t,i) = \abs{E(t,i)}e^{\delta_(t,i)}, i \in \mathcal{N}_G$} and the corresponding vector as $\bm{e}_{t} \in  \mathbb{C}^{\abs{ \mathcal{N}_G }}$ so that \response{$\left[\bm{e}_{t}\right]_i = E(t,i) $}. The current at generator bus in \eqref{eq:Ohm}, $\bm{i}^g_{t}$, is obtained as the multiplication of generator admittance and the difference in voltage at the internal bus and the generator bus \cite{GraphPowerSystem} :
\begin{equation}\label{eq:gen-current}
\bm{i}^g_{t}=
\texttt{diag}(\bm y_g)\left(\bm{e}_{t}-\bm{v}^{g}_{t}\right)
\end{equation}
\response{As mentioned in Section. \ref{sec:introduction}, the generators respond to electric load in the grid. 
In order to model the generators response, a commonly used approximation is that at the load buses $\ell \in \mathcal{N}_L$ are slowly varying admittances \cite{GraphPowerSystem}. We denote them as  $\bm{y}_{\ell}(t) \in \mathbb{C}^{N}$.}
\section{ Graph Signal Processing for the grid}\label{sec:gridGSP}
Having described the relevant GSP concepts and introduced grid quantities and parameters of interest, we are ready to introduce the Grid-GSP framework\footnote{Our preliminary GSP modeling effort can be found in \cite{DSW2019}}. Firstly, we define the GSO for the grid, then support the definition by introducing the graph-filter model that justifies it, and finally characterize its temporal dynamics. 
All of the above yields a GSP generative model for the voltage phasor measurements as a low-pass GSP model, as detailed next.  
 
\subsection{Grid graph generative model} \label{subsec:GSO} \label{subsec:grid_GSP_generative_model}
  Grid-GSP for voltage phasors data  relies on the following definitions:
 \begin{defn}\label{def.gso}
 The {\bf graph shift operator} (GSO) is a complex symmetric matrix equal to a diagonal perturbation of the system admittance matrix with generator admittance values,
\begin{align}
\mathbf{S} \triangleq  \bm{Y} + \texttt{diag} \left( \begin{bmatrix}
\bm{y}_{g} + \bm{y}_{sh}^{g} \\
\bm{y}_{sh}^{\ell}
\end{bmatrix}\right) \label{eq:grid_GSO}
\end{align}
 \end{defn}
 From the definition of the GSO it follows that:
 \begin{defn}\label{def.gft}
 	The grid {\bf Graph Fourier Transform} (GFT) basis for voltage phasors is the orthogonal matrix $\mathbf{U}$ given by the eigenvalue decomposition of the GSO in Definition \ref{def.gso}:
 	\begin{align}
 	\mathbf{S} = \mathbf{U} \bm{\Lambda} \mathbf{U}^{\top}, ~~ \abs{\lambda_{\text{min}}} > 0
 	\end{align}
 \end{defn} 
Here, the GSO $ \mathbf{S}$   is a \textit{complex-symmetric} matrix that has the same support as the electric-grid graph Laplacian as $\bm{Y}$ with the diagonal addition of  generator admittances. 
Note that unlike the graph Laplacian, this GSO is \textit{invertible}, $\abs{\lambda_{\text{min}}} > 0$.  
Even when shunt elements $\bm{y}_{sh}^{g}, \bm{y}_{sh}^{\ell}$ are ignored as conventionally done to solve power-flow problems in power systems,  a diagonal term with the generator admittances $\bm{y}_{g}$ that is added to the principal diagonal of $\bm{Y} $, makes the GSO $\mathbf{S}$ invertible \response{(see \eqref{eq:grid_GSO})}. 
\par \response{With the GSO $\mathbf{S}$ is defined as in \eqref{eq:grid_GSO}, one can rewrite \eqref{eq:Ohm} and substitute for $\bm{i}^g_{t}$ from \eqref{eq:gen-current} :
\begin{align}
 &\left(\!\bm Y \!+\!
 \texttt{diag}\! \left( \begin{bmatrix}
 \bm{y}_{sh}^{g} \\
 \bm{y}_{sh}^{\ell}
 \end{bmatrix}\right)\!\right)\bm{v}_{t} \! =\!  \begin{bmatrix}
 \texttt{diag}(\bm y_g)\left(\bm{e}_{t}-\bm{v}^{g}_{t}\right) \\
 \bm{i}^{\ell}_t
 \end{bmatrix} 
\end{align}
\response{From now on with slight abuse of notation we denote  $\bm{v}_t$ as voltage phasor measurements that are noisy therefore we add measurement noise $\eta_t$ which yields the following equation,}
 \begin{align}
 \boxed{\bm{v}_{t} =  \mathcal{H}( \mathbf{S}) \begin{bmatrix}
 	\texttt{diag}(\bm y_g)\bm{e}_t\\
 	\bm{i}^{\ell}_t
 	\end{bmatrix}
 	+ \bm{\eta}_{t}}
 \label{eq:LPmodel}
 \end{align}
The  Grid-GSP generative model for voltage phasor measurements $\bm{v}_{t}$ is given by \eqref{eq:LPmodel}. The linear shift-invariant graph filter is $\mathcal{H}( \mathbf{S})=    \mathbf{S}^{-1}$.}
\begin{rem}
	Shift-invariance  of $\mathcal{H}( \mathbf{S}) = \mathbf{S}^{-1}$ can be directly verified from \eqref{eq:shift_invariance} i.e. $\mathcal{H}( \mathbf{S}) \mathbf{S}= \mathbf{S}  \mathcal{H}( \mathbf{S}) = \mathbb{I}$. \response{Since $\mathcal{H}( \mathbf{S})$ is shift invariant, it can be expressed as a matrix polynomial in $\mathbf{S}$ (Theorem 1 in \cite{sandryhaila2013discrete}). 
	Also,  $\mathcal{H}( \mathbf{S}) = \mathbf{S}^{-1}$ can be written as in \eqref{eq:graph_filter} where coefficients $h_k$  can be determined by the application of Cayley-Hamilton theorem for inverse matrices \cite{decell1965application}. }	
\end{rem}
$\mathcal{H}( \mathbf{S})$  is approximately a low-pass graph filter \cite{Ramakrishna2020} due to the inversion of GSO since the graph frequency  response  of the filter can be written from \eqref{eq:graph_filter_freq} as $\texttt{diag}\left(\tilde{\bm{h}}\right) = \bm{\Lambda}^{-1}$. This implies that as the graph frequency decreases, the magnitude of the filter response declines. 
\response{More importantly, since generic power grids tend to be organized as communities system admittance matrix $\bm{Y}$ tends to be sparse \cite{sato1963techniques}. Therefore the GSO $\mathbf{S}$ has a high condition number and the graph frequency response of $\mathcal{H}( \mathbf{S})$ is such that it tapers off after a certain $\lambda_k$.}

 To visualize this more explicitly,  consider 	$\bm{\Lambda}_\mathcal{K}$ to be the diagonal matrix with entries $\lambda_i,~i\in\mathcal{K}=\{1,\ldots,k\}$. Define a low-pass filter with $k$ frequency components and consequently the voltage phasor measurements as
\begin{align}
\mathcal{H}_{k}( \mathbf{S}) &\triangleq\! {\mathbf U} ~\texttt{diag}\left(\tilde{\bm{h}}_{k}\right) {\mathbf U}^{\top}, 
	\left[\tilde{\bm{h}}_{k}\right]_{i} \!\!=\!\! \begin{cases}
			\lambda_{i}^{-1}, ~i \in\mathcal{K} \\
			0, ~~ \text{else}
	\end{cases} \nonumber\\
{\bm v}_t &\approx \mathcal{H}_{k}( \mathbf{S}) \begin{bmatrix}
 \texttt{diag}(\bm y_g)\bm{e}_t\\
 \bm{i}^{\ell}_t
\end{bmatrix}  +{\bm \eta}_t, \label{eq:generative_model}
\end{align}
where $\mathcal{H}_{k}( \mathbf{S})$ will  represent the principal subspace of the voltage phasors whose dimensionality is the number of graph-frequencies  $|\mathcal{K}|$.  Therefore \eqref{eq:generative_model} defines the low-dimensional generative model for quasi-steady state voltage phasor measurements.
\response{The  error term $ \bm{\eta}_{t}$ now also captures modeling approximation. }
To provide insights on the temporal dynamics of the voltage phasors, we need to capture the structure of the excitation term. As a matter of fact, $\bm{e}_{t}$ and $\bm{i}^{\ell}_{t}$, have different dynamics, as discussed  in the subsequent subsections.
\subsection{A GSP model for generator dynamics: $\bm{e}_{t}$}\label{subsec:generation}
The excitation term corresponding to generator currents has elements as $\left[\bm{e}_{t}\right]_{i}$ coming from each generator $i \in G$. We illustrate a non-linear dynamical model for the generators internal voltages, namely $\bm{e}_t \in \mathbb{C}^{\mathcal{N}_G}$ \response{utilizing a GF-AR(2) graph temporal filter from Section \ref{subsec:GSP_time}.}
The model is inspired by the classical swing equations\cite{Dorfler2013,kron_reduction_power_network} that describes the coupled dynamics of the generators phase angles, ${\delta}_{i}(t), ~ i \in G$ and the resulting variation in frequency, $ {\omega}_{i}(t) \triangleq \dot{\delta_{i}} - \omega_{0}$ where ${\omega}_{0} = 2\pi f_{0} $ with $f_{0}$ being the grid frequency ($50$ or $60$ Hz). 
\par Our model, relies on two steps. 
First, we model the dynamics of a signal $\bm{x}_t$ obtained through the following non-linear transformation of the internal generator voltages:
\begin{align}
\bm{x}_t &\triangleq (\texttt{diag}(\bm{m}))^{\frac 1 2}\ln(\bm{e}_t) \\
\bm{\delta}_t&= (\texttt{diag}(\bm{m}))^{-\frac 1 2} \Im \{ \bm{x}_t \}, \abs{\bm{e}}_t \!=\! (\texttt{diag}(\bm{m}))^{-\frac 1 2} \Re \{ \bm{x}_t \} \nonumber
\end{align}
where the vector $\bm{m}$ entries are the so-called generators masses, $\bm{\delta}_t$ are the generators angles that appear in the swing equations and $\abs{\bm{e}}_t $ are internal generator voltage magnitudes. 
Second, like in the swing equations, to describe the generators interactions, we resort to a Kron-reduction \cite{Dorfler2013,kron_reduction_power_network} of the network, in which generators are all adjacent.  
To define this generator-only network and the corresponding GSO, consider the following admittance matrix, $\bm{Y}_{\text{all}}$ that describes the network topology consisting of the generator internal buses, generator buses and non-generator buses like done in \cite{kron_reduction_power_network}:
	\begin{align}
	&\bm{Y}_{\text{all}}\!=\! \begin{bmatrix}
	\texttt{diag}\left(\bm{y}_{g} \!+\! \bm{y}_{sh}^{g} \right) & -\begin{bmatrix}
	\!\texttt{diag}\left(\bm{y}_{g}\! +\! \bm{y}_{sh}^{g} \right) \!&\! \bm{0}
	\end{bmatrix}\\
	-\begin{bmatrix}
	\!\texttt{diag}\left(\bm{y}_{g} \!+\! \bm{y}_{sh}^{g} \right) \!&\! \bm{0}
	\end{bmatrix}^{\top}& \mathbf{S}_{\text{ci}}\\
	\end{bmatrix} \nonumber \\
	&\mathbf{S}_{\text{ci}} \triangleq \mathbf{S}+  \!\texttt{diag}\left( \begin{bmatrix} \bm 0 &
	   \bm{y}_{\ell}^{\top} \end{bmatrix}\right) \nonumber
	\end{align}
In order to model $\bm{Y}_{\text{all}}$, it is assumed that the loads are varying very slowly in time i.e. $\bm{y}_{\ell}(t) \approx \bm{y}_{\ell} ~~\forall t$.
Then, let us denote by $\mathfrak{Sh}(\bm{A}, \bm{B} )$ the Schur complement of block $\bm{B}$ of matrix $\bm A$.
We compute the Schur complement of block $\mathbf{S}_{\text{ci}}$ of the matrix $\bm{Y}_{\text{all}}$ which is nothing but Kron reduction.
The Schur complement of the $\mathbf{S}_{\text{ci}}$ in $\bm{Y}_{\text{all}}$ has two contributions:
\begin{equation}
\mathfrak{Sh}(\bm{Y}_{\text{all}}, \mathbf{S}_{\text{ci}}   )=j\bm{Y}_{\text{red}}+{\bm E}_{\text{red}}
\end{equation}
where ${\bm E}_{\text{red}}$ is a real diagonal dominated matrix, and the imaginary part $\bm{Y}_{\text{red}}$ has the structure of a graph Laplacian. 
The proposed dynamical model for the graph signal $\bm x_t$ relies on the following definition for the GSO of the Kron-reduced generator-only graph:
\begin{defn}
A GSO is defined for the  Kron-reduced generator only network as
\begin{equation}
    \mathbf{S}_{\text{red}}=
    (\texttt{diag}(\bm m))^{-\frac 1 2}
    \bm{Y}_{\text{red}}(
    \texttt{diag}(\bm m))^{-\frac 1 2} \in \mathbb{R}^{\abs{\mathcal{N}_G}}
\end{equation}
with the following eigenvalue decomposition,
\begin{align}
\mathbf{S}_{\text{red}} = \mathbf{U}_{\text{red}}  \bm{\Lambda}_{\text{red}}  \mathbf{U}_{\text{red}}^{\top} 
\end{align}
and the orthonormal GFT basis being $\mathbf{U}_{\text{red}}$. 
\end{defn}
We introduce the GSP based dynamical model for the complex-valued generator internal voltages $\bm{e}_{t}$ \response{via graph temporal filter GF-AR (2)} as follows,
\begin{framed}
{\bf GSP-based dynamics for generator internal voltages}
\begin{equation}
    \bm{e}_t=\exp\left((\texttt{diag}(\bm{m}))^{-\frac 1 2}\bm{x}_t\right)
\end{equation}
where $\bm{x}_t$ is a \response{GF-AR (2) process}, i.e. the $z$-GFT $\tilde{\bm X}(z)$ satisfies the following:
\begin{align}
   \texttt{diag}\left(
    \tilde{\bm a}(z)
    \right)
    \tilde{\bm{X}}(z)&=
    \tilde{\bm{W}}(z) \\
    \tilde{\bm a}(z) &=1-\tilde{\bm a}_1 z^{-1}-\tilde{\bm a}_2 z^{-2}
\end{align}
\vspace{-2em}
\end{framed}
The GSP based dynamical model takes inspiration from swing equation for generator angles  and we empirically choose to model  generator internal voltage magnitude $\abs{\bm{e}}_t$ also using a GF-AR (2) model although in most power system models, the dynamics of the amplitudes of the generators internal voltages are typically ignored. 
The swing equation for the generators angles \cite{ZhuFreqResp} are a key tool for power systems dynamical analysis:  
\begin{align}
\texttt{diag}\left(\bm{m} \right)\ddot{\bm{\delta}} + \texttt{diag}\left(\bm{d} \right)\dot{\bm{\delta}} =  \overline{\bm{w}} - \bm{Y}_{\text{red}}  \bm{\delta}, \label{eq:diff_eqn_delta1}
\end{align}
where $\bm{m}$ are the generators masses, introduced previously,  $\bm{d}$ are the damping coefficients of generators (often neglected) and $\overline{\bm{w}} - \bm{Y}_{\text{red}}  \bm{\delta}$ is the imbalance between the electrical and mechanical power that triggers the change in generator angular velocity and acceleration. 
Note that $\Im\{\bm x\}=\left(\texttt{diag}\left(\bm{m} \right)\right)^{\frac 1 2} \bm{\delta}$. 
We can manipulate \eqref{eq:diff_eqn_delta1} to prove the following:
\begin{prop}
Let $\bm{w}$ be such that $\Im\{\bm{w}\}=(\texttt{diag}(\bm{m}))^{-\frac 1 2}\overline{\bm{w}}$.
    Using the approximation: 
    \[\texttt{diag}\left(\bm{d} \right)\texttt{diag}\left(\bm{m} \right)^{-1}\approx \chi \mathbb{I}, \] i.e. the homogeneous simplification as in \cite{ZhuFreqResp,Paganini2020}, 
    the dynamics of $\Im\{\bm x_t\}$ can be justified with an GF-AR (2) model with GSO $\mathbf{S}_{\text{red}}$: 
    \begin{align}
        &\Im\{\bm x_t\}-{\cal A}_1(\mathbf{S}_{\text{red}})\Im\{\bm x_{t-1}\}
        -{\cal A}_2(\mathbf{S}_{\text{red}})\Im\{\bm x_{t-2}\}\approx\Im\{\bm{w}\},\nonumber\\
        &{\cal A}_1(\mathbf{S}_{\text{red}}):=
        (2\!-\!\chi)\mathbb{I}-\mathbf{S}_{\text{red}},~~{\cal A}_2(\mathbf{S}_{\text{red}}):=(\chi\!-\!1)\mathbb{I} \label{eq:swing_eq}
    \end{align}
\end{prop}
\begin{proof}
Simple algebra on \eqref{eq:diff_eqn_delta1} allows us to recast the equations in the following form:
\begin{align}
\Im\{\ddot{\bm x}\} + \chi\Im\{\dot{\bm x}\}=  \texttt{diag}\left(\bm{m} \right)^{-\frac 1 2}\overline{\bm{w}} - \mathbf{S}_{\text{red}}  \Im\{\bm x\}. \label{eq:diff_eqn_Ix}
\end{align}
Assuming that the sampling rate is fast enough, and normalizing it to 1, the finite difference approximations for the derivatives are $\dot{\bm x}\approx \bm x_t-\bm x_{t-1},~~
    \ddot{\bm x}\approx \bm x_{t+1}-2\bm x_t+\bm x_{t-1}$ and can be used 
to obtain AR (2) GF equations for the samples $\Im\{\bm x_t\}$ in \eqref{eq:swing_eq}. 
\end{proof}
The model that we introduce is simply extending the GF-AR (2) model to capture both the real and imaginary part of $\bm{x}_t$  i.e. internal generator voltage magnitudes$\abs{\bm{e}}_t$ and angles $\bm{\delta}_t$ respectively and suggesting to search the $2|{\cal N}_{\cal G}|$ parameters to fit the model with $\tilde{\bm a}_1, \tilde{\bm a}_2$ rather than exploring a general MIMO filter response. 
For simplicity of representation, we write the dynamical equation for $\bm{x}_{t}$ in the GF domain,
\begin{align}
\tilde{\bm{x}}_{t} = \texttt{diag}\left( \tilde{\bm{a}}_{1} \right) \tilde{\bm{x}}_{t-1} + \texttt{diag}\left( \tilde{\bm{a}}_{2} \right) \tilde{\bm{x}}_{t-2} + \tilde{\bm{w}}_{t} \label{eq:time-GF}
\end{align}
such that the impulse response of the filter  at graph-frequency $\lambda_{\text{red},i}$ is defined by $\left[\tilde{\bm{a}}_{1}\right]_{i}, \left[\tilde{\bm{a}}_{2}\right]_{i} $. 
\begin{figure}[h]
	\centering
	\includegraphics[width= \columnwidth]{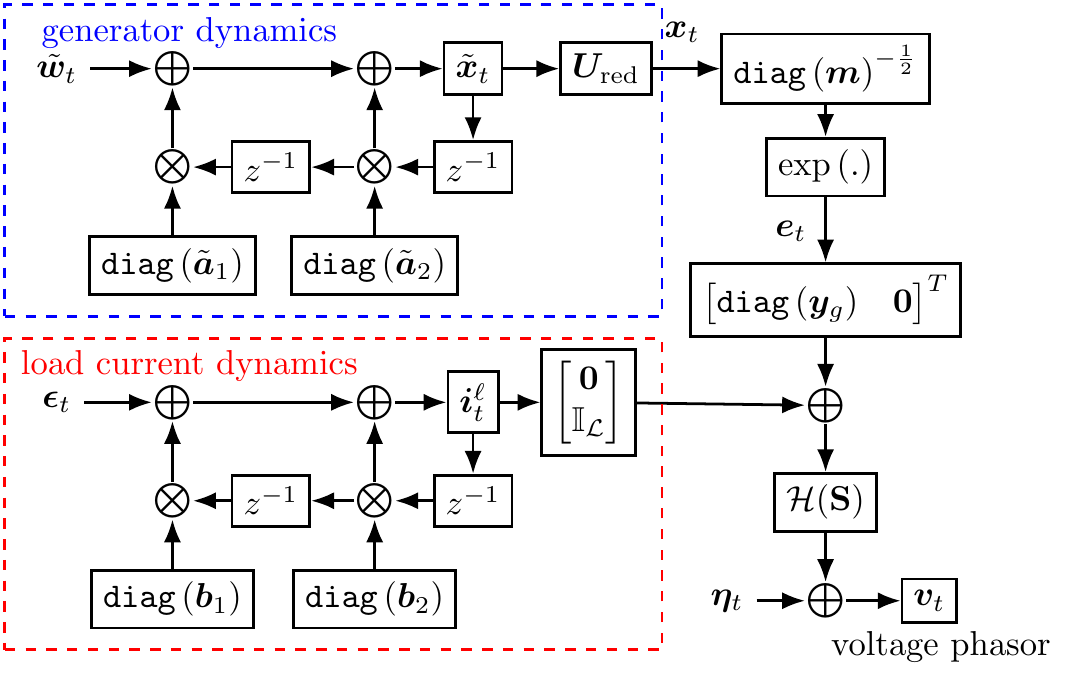}
	\caption{Block diagram showing generative model for voltage phasor measurements. }
	\label{fig:blk_diagram}
\end{figure}
Note from \eqref{eq:GF_h} that $\left[\tilde{\bm{a}}_{1}\right]_{i}, \left[\tilde{\bm{a}}_{2}\right]_{i} $ can be written as polynomials in graph frequency $\lambda_{\text{red},i}$,
\begin{align}
\textstyle \left[\tilde{\bm{a}}_{1}\right]_{i} = \sum_{k=0}^{K_1-1} a_{k,1} \lambda_{\text{red},i}^{k} ,  ~\left[\tilde{\bm{a}}_{2}\right]_{i} = \sum_{k=0}^{K_2-1} a_{k,2} \lambda_{\text{red},i}^{k},
\end{align}
which characterizes the poles of the AR system using graph frequencies. From \eqref{eq:swing_eq}, a first order polynomial may suffice to characterize $\left[\tilde{\bm{a}}_{1}\right]_{i}, \left[\tilde{\bm{a}}_{2}\right]_{i} $ in most cases. 
\subsection{Load dynamics: $\bm{i}^{\ell}_{t} $}\label{subsec:load_dynamics}
There are several papers in the literature that deal with load forecasting and modeling \cite{gao2012dynamic}. We adopt a simple AR-2 model per node or load bus to describe the dynamics of the load, 
\begin{align}
\bm{i}^{\ell}_{t} = \texttt{diag} \left( \bm{b}_{1} \right) \bm{i}^{\ell}_{t-1} + \texttt{diag} \left( \bm{b}_{2} \right) \bm{i}^{\ell}_{t-2} + \bm{\epsilon}_{t}
\end{align}
where parameters $\bm{b}_{1}, \bm{b}_{2}$ are estimated load data time series. 
The block diagram in Fig.\ref{fig:blk_diagram} summarize our modeling efforts.
\par The unique nature of voltage phasor measurements allows us to describe a similar model for any subset of measurements on a graph. This is discussed next. 
\subsection{Low-pass property of down-sampled voltage graph signal}\label{subsec:downsampled_signal}
 Let $\bm{v}_{\mathcal{M}}$ (time index $t$ is ignored for simplicity) be the down-sampled voltage graph signal where $\mathcal{M} \subset \mathcal{N} $ is the set of node indices at which measurements are available.
 \response{ It can be shown that \textit{any} down-sampled  graph signal with arbitrary graph frequency response  is low-pass in the reduced-graph  frequency domain.  
 It suggests that one can utilize all the methods for low-pass graph signals onto  down-sampled versions of the graph signal as well.  This is summarized in lemma \ref{prop:Kron_LP} below.
\begin{lem}
\label{prop:Kron_LP}	
Let $\bm{v}_{\mathcal{M}}$ be any graph signal  down-sampled in the vertex-domain with $\abs{\mathcal{M}}$ samples. Let the GSO defined with respect to the full graph $\mathbf{S}$ be invertible. Then, with the GSO defined with respect to the reduced-graph of $\mathcal{M}$ vertices  as $\mathbf{S}_{\text{red},\mathcal{M}} $,
 graph signal $\bm{v}_{\mathcal{M}}$ is the output of  low-pass graph filter $\mathcal{H}\left(  \mathbf{S}_{\text{red},\mathcal{M}} \right) \triangleq \mathbf{S}_{\text{red},\mathcal{M}}^{-1}$
\begin{align}
\bm{v}_{\mathcal{M}} = \mathcal{H}\left(  \mathbf{S}_{\text{red},\mathcal{M}} \right)  \bm{\varphi} 
\end{align}
where the GSO for the reduced-graph is given by Kron-reduction of $\mathbf{S}$,
$\mathbf{S}_{\text{red},\mathcal{M}} = \mathfrak{Sh}\left( \mathbf{S}, \mathbf{S}_{\mathcal{M}^{c}\mathcal{M}^{c}} \right)$.
\end{lem}
\begin{proof}
	Consider a graph signal $\bm{v}$ with arbitrary graph frequency response  with respect to  GSO $\mathbf{S}$,
	\begin{align}
	\bm{v}  = \mathcal{H} \left(  \mathbf{S}  \right) \bm{x}  = \mathbf{S}^{-1} \left(\mathbf{S} \mathcal{H} \left(  \mathbf{S}  \right) \bm{x}\right) \label{eq:arbitrary}
	\end{align}
	The GSO $\mathbf{S}$ is rewritten in a $2 \times 2 $ block form 
	\begin{align}
	\mathbf{S} = \begin{bmatrix}
	\mathbf{S}_{\mathcal{M}\mathcal{M}} & \mathbf{S}_{\mathcal{M}\mathcal{M}^{c}} \\
	\mathbf{S}_{\mathcal{M}\mathcal{M}^{c}}^{\top} & \mathbf{S}_{\mathcal{M}^{c}\mathcal{M}^{c}}
	\end{bmatrix},
	\end{align}
	and $\mathbf{S}^{-1}$ can be written using  inverse formula for block matrices. When graph signal $\bm{v}$ is down-sampled, 
only $\mathcal{M}$ rows are considered on both sides of \eqref{eq:arbitrary}. Thus we have, 
	\begin{align}
	\bm{v}_{\mathcal{M}} &=  \mathbf{S}_{\text{red},\mathcal{M}}^{-1} \overbrace{\begin{bmatrix}
	\mathbb{I}_{\abs{\mathcal{M}}} & - \mathbf{S}_{\mathcal{M}\mathcal{M}^{c}} \mathbf{S}_{\mathcal{M}^{c}\mathcal{M}^{c}}^{-1}  
	\end{bmatrix}  \left(\mathbf{S} \mathcal{H} \left(  \mathbf{S}  \right) \bm{x}\right) }^{\bm{\varphi}}\label{eq:kron_red_GSO_for_placement} 
	\end{align}
	where $\mathbf{S}_{\text{red},\mathcal{M}}$ is the Schur complement of the block $\mathbf{S}_{\mathcal{M}^{c}\mathcal{M}^{c}}$ in the GSO $\mathbf{S}$ i.e., 
	\begin{align}
	  \!\!\! \mathbf{S}_{\text{red},\mathcal{M}}\!=\!\mathfrak{Sh}\!\left( \mathbf{S}, \!\mathbf{S}_{\mathcal{M}^{c}\mathcal{M}^{c}}\! \right) \!=\!\mathbf{S}_{\mathcal{M}\mathcal{M}}\!-\!\mathbf{S}_{\mathcal{M}\mathcal{M}^{c}}\mathbf{S}_{\mathcal{M}^{c}\mathcal{M}^{c}}^{-1}  \mathbf{S}_{\mathcal{M}\mathcal{M}^{c}}^{\top} ~~~~ \qedhere 
	\end{align}  
\end{proof} }
\response{
Lemma \ref{prop:Kron_LP}  translates to an interesting self-similarity/fractional property  for voltage graph signals in that the down-sampled version $\bm{v}_{\mathcal{M}}$ is still a \textit{low-pass graph signal}. 
The self-similarity is due to $ \mathcal{H} \left(  \mathbf{S}  \right) = \mathbb{I}$. In summary, for voltage graph signals,  
	\begin{align}
	\bm{v} = \mathbf{S}^{-1} \bm{i}, ~~ \bm{v}_{\mathcal{M}} = \mathbf{S}_{\text{red},\mathcal{M}}^{-1} \left(\bm{i}_{\mathcal{M}} \!-\!  \mathbf{S}_{\mathcal{M}\mathcal{M}^{c}} \mathbf{S}_{\mathcal{M}^{c}\mathcal{M}^{c}}^{-1} \bm{i}_{\mathcal{M}^{c}} \right)
	\end{align}
In the power grid, this property has been illustrated empirically in several papers \cite{Wang2015,Wang2017} that highlight low-dimensionality of   measurements from a subset of buses. 
Although the reduced-graph is denser compared to the original graph, it still helps to infer faults or events that occurred in a subset of nodes where sensors are not installed as long as correct placement strategies are devised i.e. that of choosing the subset $\mathcal{M}$. Work in \cite{JSAC} explored the optimal placement for fault localization in the under-sampled regime and also made connections with GSP theory. }
\section{Revisiting algorithms from GSP for PMU data}\label{sec:GSP_algos_PMU_data}
In this section we  study some of the implications Grid-GSP has while understanding sampling, optimal placement of measurement devices in power systems, interpolation of missing samples and network inference. 
The underlying generative model responsible for low-rank nature of  data that has been established in the previous section helps explaining the success that many past works, such as \cite{Wang2017,Dahal2012,ADMM_PMU}, have attained in recovering missing PMU data using matrix completion methods.
The low-pass nature of the voltage graph signals discussed in Section \ref{sec:gridGSP} provides the theoretical underpinning that support the arguments made in the literature. 
\vspace{-0.75em}
\subsection{Sampling and recovery of grid-graph signals}\label{subsec:sampling_PMU}
 From the approximation in \eqref{eq:generative_model} we see that voltage  graph signals have graph frequency content that drops as $\lambda_{k}$ grows. This characteristic renders the signal approximately  band-limited in the GFT domain  \cite{tsitsvero2016signals} which means that there is a cut-off frequency $\lambda_{k}$ such that frequency content corresponding to $\lambda_{k+1}$ and higher is negligible. Let the GFT basis corresponding to the first dominant $k$ graph frequecies be $\mathbf{U}_{\mathcal{K}}$. 
 The {\it bandlimiting operator} is, $\bm{\mathcal{B}}_{\mathcal{K}} = \mathbf{U}_{\mathcal{K}} \mathbf{U}^{\top}_{\mathcal{K}} \in \mathbb{C}^{\abs{\mathcal{N}} \times \mathcal{K}}$ and the low frequency component of $\bm{v}_{t}$ is:
\begin{align}
\bm{\mathcal{B}}_{\mathcal{K}} \bm{v}_{t} = \mathbf{U}_{\mathcal{K}} \mathbf{U}^{\top}_{\mathcal{K}}  \bm{v}_{t}
\end{align}
Similarly, a vertex limiting operator (with $\abs{\mathcal{M}}$) vertices is $\bm{\mathcal{D}}_{\mathcal{M}} = \bm{P}_{{\mathcal{M}}}\bm{P}^{\top}_{{\mathcal{M}}}$ where $\bm{P}_{{\mathcal{M}}}$ has columns that are coordinate vectors such that each column chooses a vertex/node. When the voltage measurements on the electrical network are from a few nodes, $i \in \mathcal{M}$ at time $t$, it can be written as $\left[\bm{v}_{t}\right]_{\mathcal{M}} = \bm{P}^{\top}_{{\mathcal{M}}}\bm{v}_{t} $. 
For reconstruction, results in \cite{tsitsvero2016signals} dictate the necessary condition be that $\abs{\mathcal{M}} \geq \abs{\mathcal{K}}$. 
In the presence of  modeling error relative to the perfect band-limited definition,  optimal sampling pattern i.e. the best placement for PMUs on the grid to minimize the worst-case reconstruction error  is closely tied to the grid topology and the model mismatch relative to a strictly band-limited graph signal \cite{Anis2016}. 
An optimal placement strategy of PMUs that minimizes the worst-case reconstruction error in the presence of model mismatch due to imperfect band-limited nature of the voltage graph signal, also known as the E-optimal design\cite{Anis2016}, is sought by  maximizing the smallest singular value, $\sigma_{\min} \left(\bm{\mathcal{D}}_{\mathcal{M}} \mathbf{U}_{\mathcal{K}}\right)$,  i.e. choose rows of $\mathbf{U}_{\mathcal{K}}$ such that they are as uncorrelated as possible and the resulting matrix has the highest condition number \cite{tsitsvero2016signals,Anis2016}.
Consider then the spatial sampling mask $\bm{\mathcal{D}}_{\mathcal{M}} = \texttt{diag}(\bm{1}_{\mathcal{M}})$ that selects $M$ locations.  
\subsubsection{Sampling}
The optimal placement of $M$ PMUs maximizes $\sigma_{\min}(\bm{\mathcal{D}}_{\mathcal{M}} {\mathbf{U}}_{\mathcal{K}})$ which amounts to choosing the rows of $\mathbf{U}_{\mathcal{K}}$ with the smallest possible coherence (as close as possible to being orthogonal). 
In \cite{tsitsvero2016signals} and references therein, a greedy method is employed to find $M$ rows from ${\mathbf{U}}_{\mathcal{K}}$ so that the least singular value is maximized. 
\par Power systems topologies exhibit naturally a {\it community structure} that is reflected in the system admittance matrix $\bm{Y}$  \cite{DSW2019} due to population density or clusters of loads. \response{It is known that a method to determine $k$ communities in a graph is to minimize the Ratio Cut \cite{HagenKahn1992} and spectral clustering  performs a relaxed  Ratio Cut minimization  via $k-$means algorithm on rows of  the eigenvectors $\mathbf{U}_{\mathcal{K}}$ \cite{spectral_clustering_tutorial}.}  Thus, choosing rows of $\mathbf{U}_{\mathcal{K}}$ to be uncorrelated is intuitively putting PMUs in different {\it graph-clusters} or  communities. This fact was also discussed in \cite{JSAC} in the context of sensor placements for fault localization. 
The PMUs sampling rate in time exceeds the needs for reconstructions in a quasi-steady state conditions by a significant margin and it is designed to help detect sharp transients in the system.
\subsubsection{Reconstruction} \label{subsubsec: reconstruction_sampling_PMU}
Voltage data samples are obtained  down-sampling in space after the optimal placement of PMUs and also uniformly down-sampling in time. At time $t$ when $\abs{\mathcal{M}}$ samples, $\left[\bm{v}_t\right]_{\mathcal{M}}$ are available, the following model applies 
\begin{align}
\left[\bm{v}_t\right]_{\mathcal{M}} \approx \bm{P}_{{\mathcal{M}}}^{\top} \mathbf{U}_{\mathcal{K}} \tilde{\bm{v}}_{t}
\end{align}
where $\tilde{\bm{v}}_{t}$ is the GFT of graph signal $\bm{v}_{t} $ 
Therefore, reconstruction in spatial domain is done via GFT basis as
\begin{align}
\hat{\bm{v}}_{t} = \mathbf{U}_{\mathcal{K}} \left(\bm{P}_{{\mathcal{M}}}^{\top} \mathbf{U}_{\mathcal{K}} \right)^{\dagger} \left[\bm{v}_t\right]_{\mathcal{M}}
\end{align}
Reconstruction in temporal domain can be done independently by \textit{up-sampling}, i.e. via the windowed inverse Fourier transform of the up-sampled signal created from uniformly time-decimated data. 
\vspace{-0.75em}
\subsection{Interpolation of missing samples}\label{subsec:interpolation}
When voltage measurements are missing or corrupted, denoising and interpolation of such data can be cast as a graph signal recovery problem by regularizing the total variation, (TV).  Overall, the problem resembles time-vertex graph signal recovery \cite{JFT}. Let $\mathbf{V} = \begin{bmatrix}
\bm{v}_{1} & \bm{v}_{2} & \dots & \bm{v}_{T}
\end{bmatrix}$ represent the voltage phasor measurements matrix collected over $T$ time instants. Let $\mathbb{P}_{\Omega}\left(\hat{\mathbf{V}}\right)$ be the set of available measurements that have samples in entries of set $\Omega$ and are noisy,
\begin{align}
\min_{\mathbf{V}} ~~ &\textstyle \norm{\mathbb{P}_{\Omega}\left(\hat{\mathbf{V}}\!-\!\mathbf{V}\right)}_{F}^{2}\! +\! c_g \sum_{t=1}^{T} \norm{\mathbf{S}\bm{v}_{t}}_{1} \nonumber \\
&+ \textstyle c_t\sum_{t=2}^{T}
\norm{\bm{v}_{t}  - \bm{v}_{t-1}}_{2}^{2} \label{eq:interp_GSO}
\end{align} 
where the two regularizing terms measure the variation in the graph and time domain and $c_g,c_t$ are the corresponding regularization constants. 
Importantly,  one can use the  GSO of the reduced graph, $\mathbf{S}_{\text{red},\mathcal{M}}$ if we only have access to a subset of measurements on the grid, $\mathcal{M}$ and employ the same formulation as in \eqref{eq:interp_GSO} for interpolation of missing samples. 
\vspace{-0.75em}
\subsection{Network inference as graph Laplacian learning} \label{subsec:network_inference}
The problem of estimation of GSO  $\mathbf{S}$ from voltage  phasor measurements can be cast as a solving a problem similar to graph Laplacian learning \cite{dong2019learning}  which seeks the GSO that minimizes the total variation of the observed voltage phasors. If current measurements $\bm{i}_{t}$ are available, then another regularization term $\norm{\mathbf{S}\bm{v}_{t} - \bm{i}_{t} }_{2}^{2}$ can be added such that Ohm's law is satisfied. 
Therefore, estimation of GSO can be accomplished by solving the following problem:
\begin{align}
    & \min_{\mathbf{S}}~~~~ \textstyle  \sum_{t=1}^{T} \norm{\mathbf{S}\bm{v}_{t}}_{1}  + \gamma \norm{\mathbf{S} - \texttt{diag}\left(  \texttt{Diag}\left(\mathbf{S}\right)  \right)}_{F}^{2} \nonumber\\
    &~~~~~~~~+ \sum_{t=1}^{T} \norm{\mathbf{S}\bm{v}_{t} - \bm{i}_{t} }_{2}^{2} \label{eq:nw_inference}\\
    &\text{subject to} ~~ \Re{\left[\mathbf{S}\right]_{i,j}} = \Re{\left[\mathbf{S}\right]_{j,i}}, \Im{\left[\mathbf{S}\right]_{i,j}} = \Im{\left[\mathbf{S}\right]_{j,i}}, i\neq j, \label{eq:complex_symmetry}\\
    & \Re {\text{Tr}\left(\mathbf{S}\right)} = \alpha \abs{\mathcal{N}}, \Im {\text{Tr}\left(\mathbf{S}\right)} = \beta \abs{\mathcal{N}} \label{eq:dominant_diag}
\end{align}
Additional constraints on the GSO  can be imposed based on the properties of 
complex-symmetry (see \eqref{eq:complex_symmetry}),  sparse off-diagonal entries via the term $\norm{\mathbf{S} - \texttt{diag}\left(  \texttt{Diag}\left(\mathbf{S}\right)  \right)}_{F}^{2}$ and dominant diagonal values (see \eqref{eq:dominant_diag}).  
Also, $\mathbf{S}$ tends to have larger imaginary values than real especially on the diagonal. 
and $\alpha, \beta >1 $ control the amplitude of real and imaginary values on the diagonal.  
As before, the problem above can be recast with down-sampled voltage graph signals to infer the Kron-reduced GSO $\mathbf{S}_{\text{red},\mathcal{M}}$ with the approximation that the term $\mathbf{S}_{\mathcal{M}\mathcal{M}^{c}} \mathbf{S}_{\mathcal{M}^{c}\mathcal{M}^{c}}^{-1} \bm{i}_{\mathcal{M}^{c}}$ in \eqref{eq:kron_red_GSO_for_placement} is treated as additive Gaussian noise.  Simulation results for network inference can be found in Section \ref{sec:results}.  
\section{Applications of Grid-GSP}\label{sec:appl_Grid_GSP}
The goal of this section is to showcase  the benefits of casting problems in the Grid-GSP framework through two exemplary applications, namely anomaly detection and data compression. The common thread between them is the use of the Grid-GFT as a tool to extract informative features from PMU data.
\subsection{Detection of FDI attacks on PMU measurements}\label{subsec:FDI_attacks}
This application is based on our preliminary work in \cite{GlobalSIP2019}. Note that, even though we cast the problem as that of FDI attacks detection, the idea can be easily extended to unveil sudden changes due to  physical events (like fault-currents, or topology changes) that similarly excite high GF content.
We assume that we have access to PMU measurements of voltage and current from the buses they are installed on.  
Let $\mathcal{A}$ be the set of available measurements where PMUs are installed and $\mathcal{U}$ be set of unavailable ones. A measurement model can be written using `state' to be the voltage as
\begin{align}
\underbrace{\begin{bmatrix}
	\hat{\bm{i}}_{\mathcal{A}} \\
	\hat{\bm{v}}_{\mathcal{A}}
	\end{bmatrix}}_{\bm{z}_t} = \underbrace{\begin{bmatrix}
	\bm{Y}_{\mathcal{A}\mathcal{A}} & \bm{Y}_{\mathcal{A}\mathcal{U}} \\
	\mathbb{I}_{\abs{\mathcal{A}  }} & \bm{0}
	\end{bmatrix}}_{\bm{H}} \underbrace{\begin{bmatrix}
	\bm{v}_{\mathcal{A}} \\
	\bm{v}_{\mathcal{U}}
	\end{bmatrix}}_{\bm{v}}  + \bm{\varepsilon}  \label{eq:meas_model}
\end{align} 
The attacker follows the strategy of FDI attack to manipulate both current and voltage on the set of malicious buses,  $ i \in \mathcal{C} \subset \mathcal{A}$ by introducing a perturbation 
\begin{align}
\delta \bm{v}^{T}_t = \begin{bmatrix}
\delta \bm{v}_{\mathcal{C} }^{T}  & \bm{0}_{\abs{\mathcal{P}}+\abs{\mathcal{U}}}^{T}
\end{bmatrix}, ~\text{such that}~
\bm{Y}_{\mathcal{P}\mathcal{C}} \delta \bm{v}_{\mathcal{C}} = \bm{0} \label{eq:delta_v}
\end{align}
where $\mathcal{P}$ is the set of honest nodes. This requires special conditions and placement, since  $\bm{Y}_{\mathcal{P}\mathcal{C}}$ is tall.
Nonetheless, since the system admittance matrix $\bm{Y}$ is generally sparse \cite{sato1963techniques}, $\bm{Y}_{\mathcal{P}\mathcal{C}}$ does not have full column-rank for a sufficient number of attackers $\mathcal{C}$ even when all the measurements are available with $\mathcal{A} = \mathcal{N}$.
Our detection problem entails deciding between the hypotheses of attack $\textit{H}_{1}$ and no attack $\textit{H}_{0}$.
To this end, we can leverage the low-dimensional generative model for the voltage graph signal that comes from \eqref{eq:generative_model}, which imposes additional constraint on the perturbation along with that in \eqref{eq:delta_v}. In short, for the attacker to be successful and undetected, she needs to have knowledge of system parameters and the graph filter with $k$ frequency components $\mathcal{H}_{k}( \mathbf{S})$. However, since the attacker does not have all this knowledge, a typical FDI attack as studied in literature is launched using \eqref{eq:delta_v}.  
Using the generative model in \eqref{eq:generative_model}, we know that under normal operating conditions in quasi-steady state, 
the received data $\bm{z}_t$ under the no-attack and attack hypotheses $\textit{H}_{0}, \textit{H}_{1}$ respectively have the structure:
\begin{align}
\hspace{-0.68em}\bm{z}_t \!=\!\! \begin{cases}
\!\textit{H}_{0}:\!\!&\hspace{-0.8em} \bm{H} \mathcal{H}_{k}( \mathbf{S})\!\! \begin{bmatrix}
\left(\texttt{diag}(\bm y_g)\bm{e}_t\right)^{\top} &
\left(\bm{i}^{\ell}_t\right)^{\top}
\end{bmatrix}^{\top} \!\!+\! \bm{\varepsilon}_t  \\
\!\textit{H}_{1}:\!\!& \hspace{-0.8em} \bm{H}\mathcal{H}_{k}( \mathbf{S})\!\! \begin{bmatrix}
\left(\texttt{diag}(\bm y_g)\bm{e}_t\right)^{\top} &
\left(\bm{i}^{\ell}_t\right)^{\top}
\end{bmatrix}^{\top}\!\!\!\! +\!\! \bm{H}\delta \bm{v}_t\!+\! \bm{\varepsilon}_t 
\end{cases}
\end{align}
Therefore, we project $\bm{z}_t$ onto the subspace orthogonal to columnspace of $ \bm{H} \mathcal{H}_{k}( \mathbf{S})$ to get a test statistic, $d({\bm z})$. 
The projector is:
\begin{align}\label{eq:FDI-filter}
\bm{\Pi}_{\perp \bm{H} \mathcal{H}_{k}( \mathbf{S}) }  \triangleq \mathbb{I} - \left[\bm{H}\mathcal{H}_{k}( \mathbf{S})\right]\left[ \bm{H}\mathcal{H}_{k}( \mathbf{S})\right]^{\dagger}
\end{align} 
and under the no attack hypothesis $\textit{H}_{0}$, energy in the orthogonal subspace
is less than when there is an attack, $\textit{H}_{1}$. This can be converted to the following test,
\begin{align}
d(\bm{z}_t)  \triangleq  \norm{ \bm{\Pi}_{\perp \bm{H} \mathcal{H}_{k}( \mathbf{S}) }  \bm{z}_t}_{2}^{2} ~~ \mathop{\gtreqless}_{\textit{H}_{0}}^{\textit{H}_{1}}~~ \tau
\end{align}
where $\tau$ is a threshold that can be chosen based on an empirical receiver operator characteristics (ROC) curve. 
Note that, since $\mathcal{H}_{k}( \mathbf{S})$ is a low pass filter, the projector $\bm{\Pi}_{\perp \bm{H} \mathcal{H}_{k}( \mathbf{S}) }$ in  \eqref{eq:FDI-filter} is filtering high graph frequencies and the detection measures the energy on such frequencies as a signature for anomalies. 

\par Isolation of compromised buses or estimate of $\delta{\bm{v}}$ can also undertaken with a similar logic. Firstly, using the assumptions in the previous section we can  solve the following regression problem to recover $\delta\bm{v}_t$, formulating a LASSO relaxation of the sparse support recovery problem:
\begin{align}
\hspace{-0.6em} \min_{\delta{\bm{v}_t}} ~ \norm{\bm{\Pi}_{\perp \bm{H} \mathcal{H}_{k}( \mathbf{S}) } ({\bm z}\!-\! \bm{H}\delta{\bm{v}})}^{2}_{2} ~~  \text{subject to} ~~  \norm{\delta{\bm{v}_t}}_{1} \!\leq\! \mu \label{eq: regression_isolation_of_nodes}
\end{align}
Constraint on the $\ell_{1}$ norm is used to incorporate the prior knowledge that the attacker has access to a few measurement buses, $\mathcal{C} \ll \mathcal{N}$. 
Note that the performance of the algorithm is also dependent on the number of graph-frequency components i.e. $k$ considered. 
\vspace{-0.75em}
\subsection{ Compression of PMU measurements} \label{subsec:compression}
 The  proposed compression algorithm leverages both \eqref{eq:LPmodel} and \eqref{eq:time-GF}. The measure of distortion we use is the mean-squared error (MSE): 
 \begin{align}
 \textstyle d(\bm{v},\hat{\bm{v}}) \triangleq (\abs{\mathcal{N}}T)^{-1} \sum_{t=0}^{T} \norm{\bm{v}_{t} - \hat{\bm{v}}_{t}}_{2}^{2}
 \end{align}
 where $T$ denotes  the time instant at which samples are stopped collecting.
   \begin{algorithm}
  	\caption{Encoding algorithm for compression}\label{alg:encoding}
  	\begin{algorithmic}[1]
  		\Require $ {\tilde{\bm{x}}}_{0}, {\tilde{\bm{x}}}_{1}, \bm{i}_{0}^{\ell},\bm{i}_{1}^{\ell}, \left\{  \bm{v}_{t} \right\}_{t=2}^{T} $
  		\For{$t=2:T$}
  		\State States $\hat{\tilde{\bm{x}}}_t, \hat{\bm{i}}_{t}^{\ell}$ from \eqref{eq:x_update_compress} and \eqref{eq: i_update_compress} respectively. 
  		\State  Voltage estimate:
  		\begin{align}
  		\hspace{-0.3em}\hat{\bm{v}}^{0}_{t} \!=\! \mathcal{H}( \mathbf{S})\! \begin{bmatrix}
  		\texttt{diag}\!\left(\bm{y}_{g}\right) \exp \left\{ \!\texttt{diag}(\bm{m})^{-\frac{1}{2}}\mathbf{U}_{\text{red}} \hat{\tilde{\bm{x}}}_{t}  \!        \right\} \\
  		\hat{\bm{i}}^{\ell}_{t}
  		\end{bmatrix} \label{eq:v_0_compress}
  		\end{align}
  		\State Compute  GFT of modeling error:  		
  		 $\tilde{\bm{\xi}}_{t} = \mathbf{U}^{\top} \left(\bm{v}_{t} - \hat{\bm{v}}^{0}_t \right)$  		
  		 \State Quantize: $\hat{\tilde{\bm{\xi}}}_{t} = \mathcal{Q}\left\{ \tilde{\bm{\xi}}_{t}   \right\}, ~ \hat{\bm{v}}_t = \hat{\bm{v}}^{0}_t +  \mathbf{U}\hat{\tilde{\bm{\xi}}}_{t}  $
  		\State Update states,
  		\begin{align}
  		&\hat{\tilde{\bm{x}}}_t \leftarrow \mathbf{U}_{\text{red}}^{\top} \texttt{diag}(\bm{m})^{\frac{1}{2}} \ln (\hat{\bm{e}}_{t}), \hat{\bm{i}}^{\ell}_{t} \leftarrow  \left[\mathbf{S} \hat{\bm{v}}_t\right]_{\mathcal{N}_L} \label{eq: state_update_compress} \\
  	&\text{where} ~\hat{\bm{e}}_{t} =  \left(\texttt{diag}(\bm{y}_g) \right)^{-1}  \left[\mathbf{S} \hat{\bm{v}}_t \right]_{\mathcal{N}_G} \nonumber  		
  		\end{align}
  		\EndFor
  		\State \textbf{end for}
  		\Ensure 
  		 $  \left\{ \hat{\tilde{\bm{\xi}}}_{t} \right\}_{t=2}^{T}$
  	\end{algorithmic}
  \end{algorithm}
 Since we have a temporal dynamical model for the evolution of voltage in time, we use differential encoding \cite{diff_enc} to quantize the residuals in both generator and load dynamics, $\tilde{\bm{w}}_{t}$ and $\bm{\epsilon}_{t}$ respectively. The  voltage at time $t$ is: 
 \vspace{-0.1em}
 \begin{align}
 &{\bm{v}}_{t} =  \mathcal{H}( \mathbf{S}) \begin{bmatrix}
 \texttt{diag}\left(\bm{y}_{g}\right)\exp \left\{ \texttt{diag}(\bm{m})^{-\frac{1}{2}}\mathbf{U}_{\text{red}} \left({\tilde{\bm{x}}}^{0}_{t} +\tilde{\bm{w}}_{t}\right)           \right\}  \\
 {\bm{i}}^{\ell,0}_{t} +  \bm{\epsilon}_{t}
 \end{bmatrix} \nonumber \\
  & \text{where}~~~~	\tilde{\bm{x}}_{t}^{0}  =  {\texttt{diag} \left( \bm{a}_{1}\right ) {\tilde{\bm{x}}}_{t-1}\! +  \texttt{diag} \left( \bm{a}_{2} \right )} {\tilde{\bm{x}}}_{t-2}  \label{eq:x_update_compress}\\
  &~~~~~~~~~~~{\bm{i}}^{\ell,0}_{t}  = 	\texttt{diag} \left( \bm{b}_{1}\right ) \hat{\bm{i}}_{t-1}^{\ell}\! +  \texttt{diag} \left( \bm{b}_{2} \right )
  		    \hat{\bm{i}}_{t-2}^{\ell}  \label{eq: i_update_compress}
 \end{align}
 Thus, $\bm{v}_{t}$ can be approximated as:
 \begin{align}
\hspace{-0.8em}{\bm{v}}_{t}\!\!\approx \! \mathcal{H}( \mathbf{S}) \!\left(\!\begin{bmatrix}
 \!\texttt{diag}\left(\!\bm{y}_{g}\!\right)\exp \left\{ \!\texttt{diag}(\bm{m})^{-\frac{1}{2}}\mathbf{U}_{\text{red}} {\tilde{\bm{x}}}^{0}_{t}          \right\}  \\
 {\bm{i}}^{\ell,0}_{t} 
 \end{bmatrix} \!\!\!+\! \bm{\xi}_t \!\!\right) \label{eq:xi_t}  
\end{align}
 Note that, the vector $\bm{\xi}_{t}$ GFT, $\mathbf{U}^{\top} \bm{\xi}_{t}$ has energy mostly in lower frequency components and is therefore an appropriate term to quantize using an optimal rate allocation. Specifically, we allocate bits to each component by setting a desired level of total distortion, applying the \textit{reverse water-filling} result \cite{CoverThomas:Book} which is optimum for a random vector  $\mathbf{U}^{\top} \bm{\xi}_{t}$ whose entries are circularly symmetric complex independent Gaussian random variables and then quantize the components accordingly.\footnote{The covariance matrix is not diagonal and ideally one would first whiten the vector $\mathbf{U}^{\top} \bm{\xi}_{t}$ and then quantize the individual components with bit-allocation akin to reverse water-filling. 
 Since the statistics of $\mathbf{U}^{\top} \bm{\xi}_{t}$ are time-varying, one has to perform the whitening transform at each time instant which is a  cumbersome operation. Therefore we make the assumption of a diagonal covariance matrix while sacrificing the benefit of modeling the underlying correlations among the random variables. } 
Then, we use the quantized vector $\mathbf{U}^{\top} \bm{\xi}_{t}$ to \textit{update} the state i.e. to estimate $\tilde{\bm{x}}_{t}$ and $\bm{i}^{\ell}_{t}$. 
Algorithms \ref{alg:encoding} and \ref{alg:decoding} describe the encoding and decoding algorithms respectively. 
  \begin{algorithm}
  	\caption{Decoding algorithm for reconstruction}\label{alg:decoding}
  	\begin{algorithmic}[1]
  		\Require $\tilde{\bm{x}}_{0}, \tilde{\bm{x}}_{1}, \bm{i}_{0}^{\ell},\bm{i}_{1}^{\ell}, $  $\left\{ \hat{\tilde{\bm{\xi}}}_{t} \right\}_{t=2}^{T}$
  			\For{$t=2:T$}
  			\State States $\hat{\tilde{\bm{x}}}_t, \hat{\bm{i}}_{t}^{\ell}$ from \eqref{eq:x_update_compress} and \eqref{eq: i_update_compress} respectively. 
  		\State Reconstruction of voltage from \eqref{eq:v_0_compress}  		
  		$\hat{\bm{v}}_{t} =  \hat{\bm{v}}_{t}^{0} +  \mathbf{U} \hat{\tilde{\bm{\xi}}}_{t}$  		
  		\State Update states $\hat{\tilde{\bm{x}}}_t, \hat{\bm{i}}_{t}^{\ell}$ from \eqref{eq: state_update_compress}
   		\EndFor
  		\State \textbf{end for}
  		\Ensure $ \left\{\hat{\bm{v}}_{t} \right\}_{t=2}^{T} $
  	\end{algorithmic}
  \end{algorithm}
 Note that the proposed scheme of compression is sequential unlike others in literature. Several corrections can be made as data is collected in time such as the update of parameters  $ \tilde{\bm{a}}_{1}, \tilde{\bm{a}}_{2},\bm{b}_{1}, \bm{b}_{2}$.
\begin{figure}
	\centering
	\includegraphics[width=0.62\columnwidth]{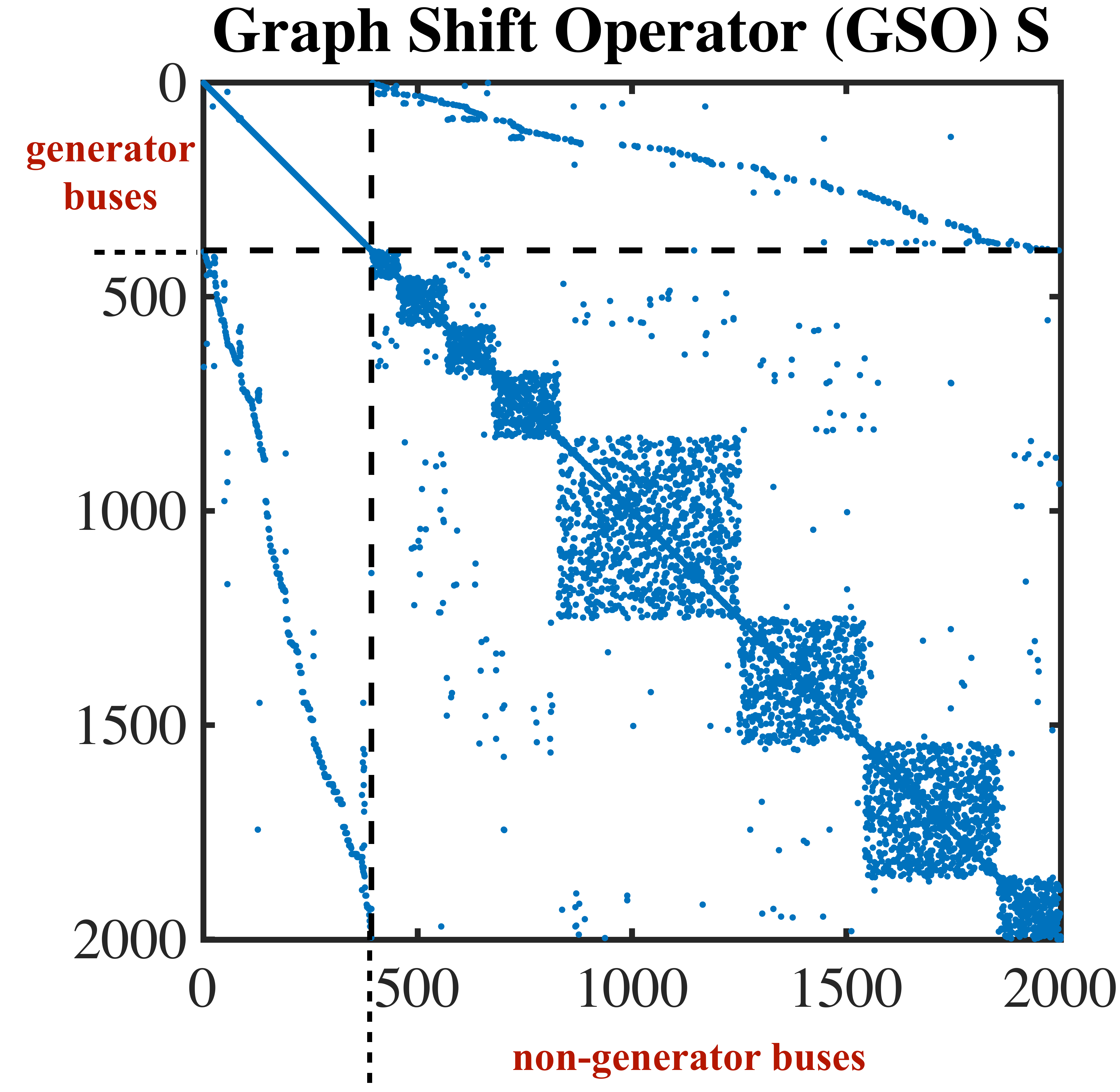}
	\caption{Support of the GSO $\mathbf{S}$ of the network.  }
	\label{fig:Ordered_Laplacian_matrix}
\end{figure}
\section{Numerical Results}\label{sec:results}
The numerical results in this section are mostly obtained using data from the synthetic ACTIVSg2000 case \cite{BirchfieldACTIVSg}, a realistic model emulating the ERCOT system, which includes $2,000$ buses-with $432$ generators and the rest non-generator buses. The ACTIVSg2000 case data include a realistic PMU data time series, in which $392$ generators are dispatched to meet variable load demand. 
The sampling rate, as for real PMUs, is $30$ samples per second. 
As all the system related parameters are known, it is easier to verify the proposed modeling strategy through the ACTIVSg2000 PMU data set. 
 Fig.~\ref{fig:Ordered_Laplacian_matrix} shows the support of the graph Laplacian or the $\bm{Y}$ matrix when ordered into generator and non-generator buses. The block-diagonal structure is notable, and is the result of the population distribution in the state of Texas, which is concentrated in $8$ metropolitan areas. 
  \par \underline{\it Grid-GSP model}: In Fig.~\ref{fig:PMU_GFT_gen}, magnitude of GFT of voltage graph signal $\bm{v}_{t}$ and the input $\bm{x}_t$ are plotted for  a single time instant  with respect to their corresponding normalized graph frequencies $\abs{\lambda_{i}}/\max_i \abs{\lambda_{i}}$ and shown in log-scale. 
  From the linear decay, it is evident that the magnitude of GFT coefficients $\abs{\tilde{\bm{v}}_{t}}$ corresponding to lower frequencies are more significant as compared to higher frequencies. 
   Similarly, the GFT of the exponent in the input, $\tilde{\bm{x}}_{t} = \mathbf{U}_{\text{red}}^{\top} \bm{x}_t$ with the generator GSO $\mathbf{S}_{\text{red}}$, is plotted with respect to the graph frequencies in Fig. \ref{fig:PMU_GFT_gen}. The decay in GFT coefficients with respect to frequency is less pronounced confirming that graph signal $\bm{x}_{t}$ is not necessarily low-pass and in general depends on the topology of the generator only network.  
  
  In Fig. \ref{fig:PMU_GFT_downsampled}, magnitude of GFT of the downsampled voltage graph signal, $ \tilde{\bm{v}}_{\mathcal{M}} = \mathbf{U}^{\top}_{\text{red},\mathcal{M}}{\bm{v}}_{\mathcal{M}} $ for $\abs{\mathcal{M}} = 867$ and $\abs{\mathcal{M}} = 261$ with two different down-sampling strategies: with PMUs placed at buses in few communities within the GSO $\mathbf{S}$  and the other being optimal placement for graph signal reconstruction. 
   \begin{figure}
 	\centering
 	\includegraphics[width=\columnwidth,height= 0.15\textheight]{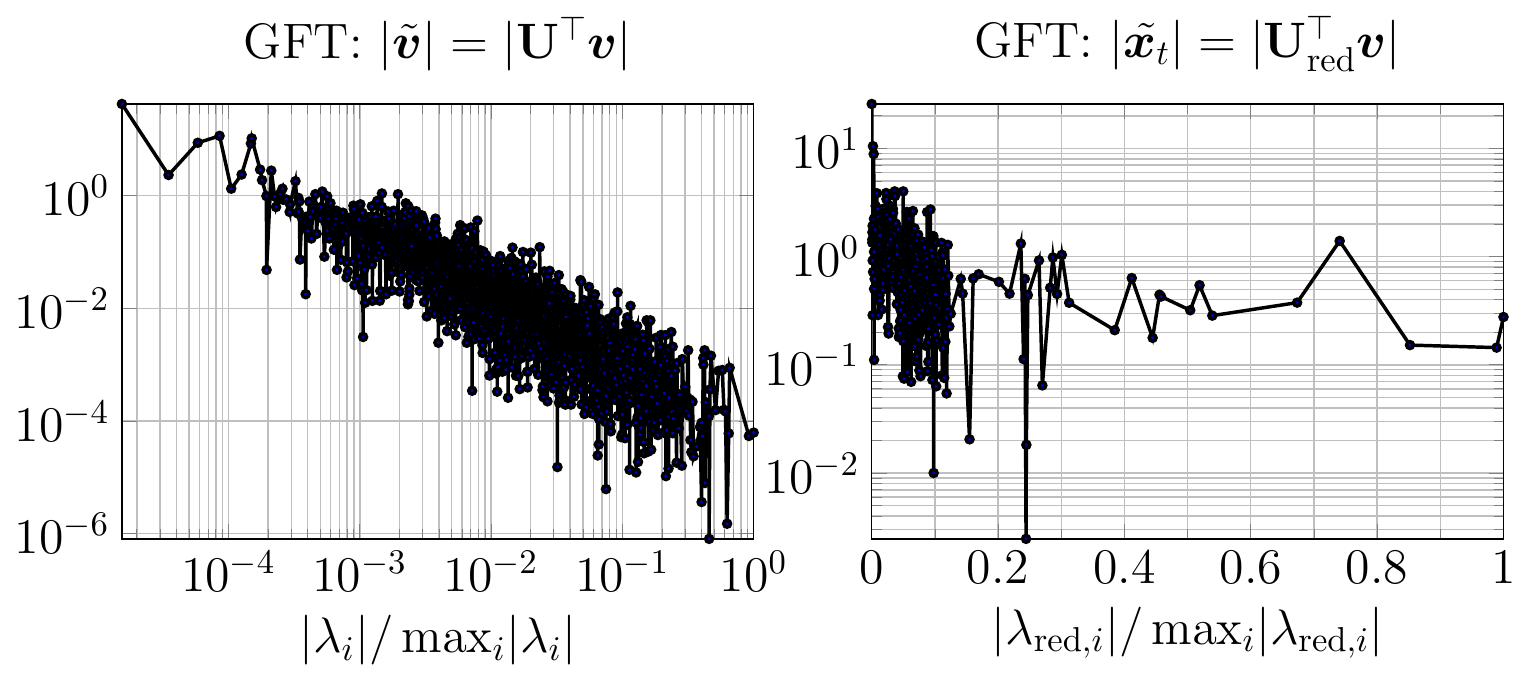}
 	\caption{Magnitude 
 	of Graph Fourier Transform (GFT) for voltage graph signal, $\abs{ \tilde{\bm{v}_{t}} }$ (left) and input to generator-only network $\abs{ \tilde{\bm{x}_{t}} }$ (right)  plotted with respect to  normalized  graph frequency. }
 	\label{fig:PMU_GFT_gen}
 \end{figure}
 \begin{figure}
 	\centering
 	\includegraphics[width=\columnwidth,height= 0.17\textheight]{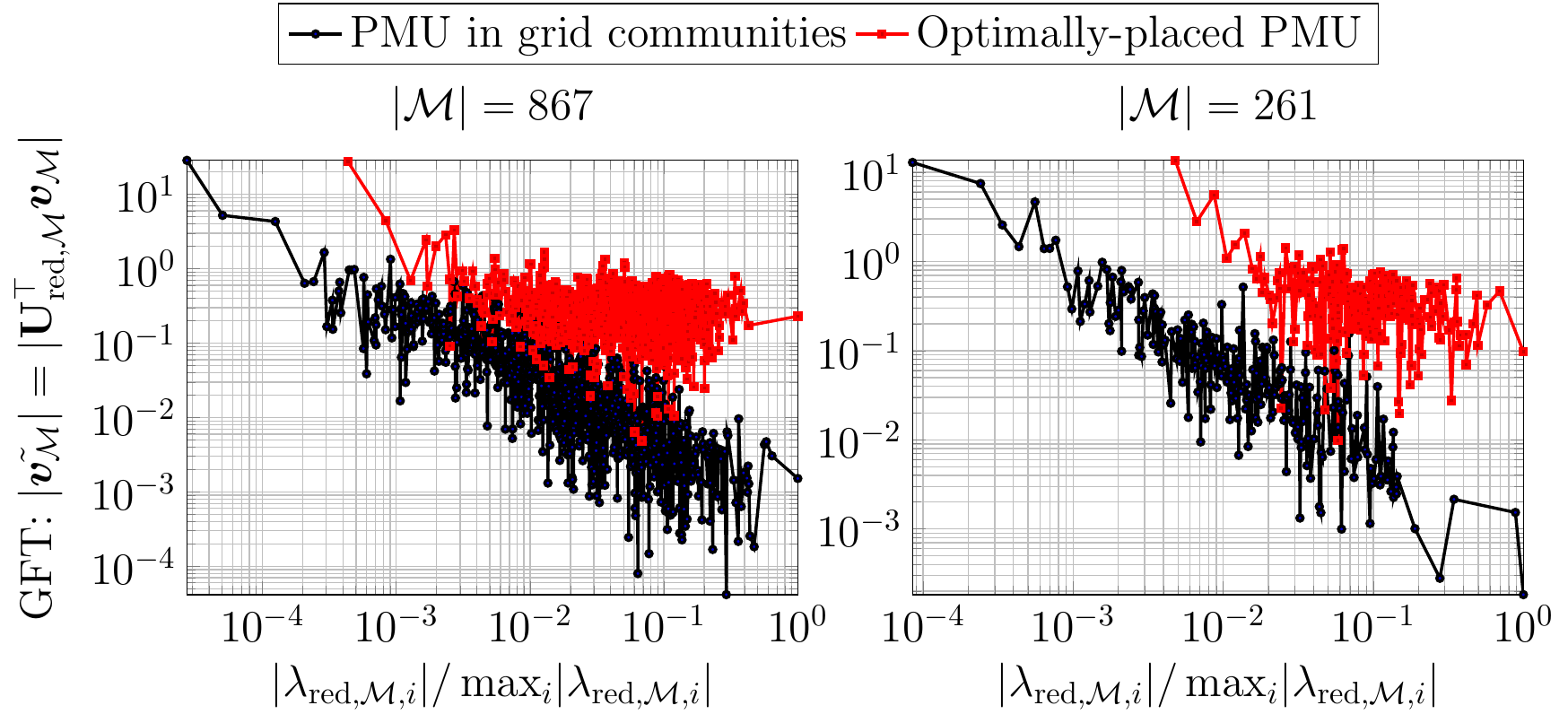}
 	\caption{Magnitude of GFT for spatially downsampled voltage graph signal, $\abs{ \tilde{\bm{v}}_{\mathcal{M}} }$  plotted with respect to normalized  graph frequency.}
 	\label{fig:PMU_GFT_downsampled}
 \end{figure}
The placement strategy has an effect on the low-pass nature of the downsampled signal. The steeper attenuation of GFT magnitude with placement strategy being community-wise is a result of loss in spatial-resolution.
\par To highlight the temporal variation in the GFT domain of input exponent, $\tilde{\bm{x}}_{t}$, a short time-series of the real and imaginary parts along with the fit of the AR model are shown in Fig. \ref{fig:tilde_x_t_time_series}.   As expected, the AR model fits well.  Fig. \ref{fig:load_dynamics_AR} shows the similar AR-2 model fit to the load current at a bus that had the highest absolute value of load. 
\begin{figure}
	\centering 
	\includegraphics[width=0.75\columnwidth]{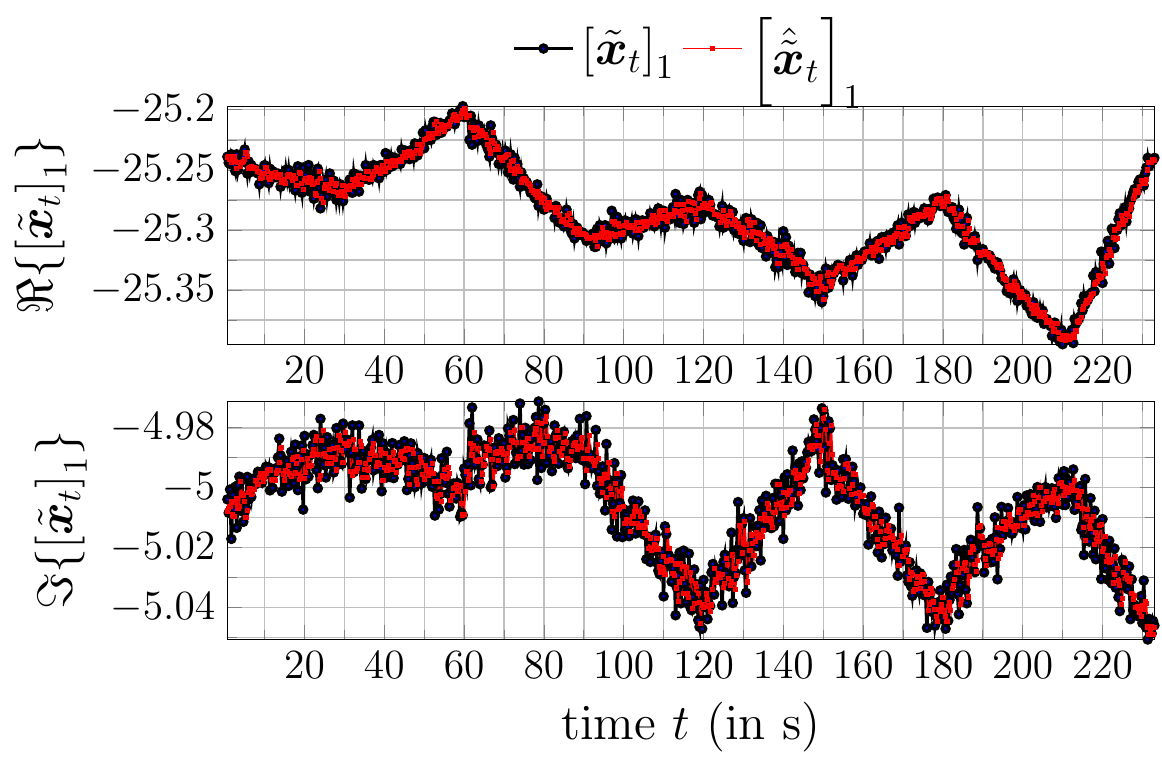}
	\caption{AR model fit to the GFT of $\bm{x}_{t}$. Component corresponding to smallest graph frequency, $\lambda_{\text{red},1}$,  $\left[\tilde{\bm{x}}_{t}\right]_1$}
	\label{fig:tilde_x_t_time_series}
\end{figure}
\begin{figure}
\centering 
\includegraphics[width=0.8\columnwidth]{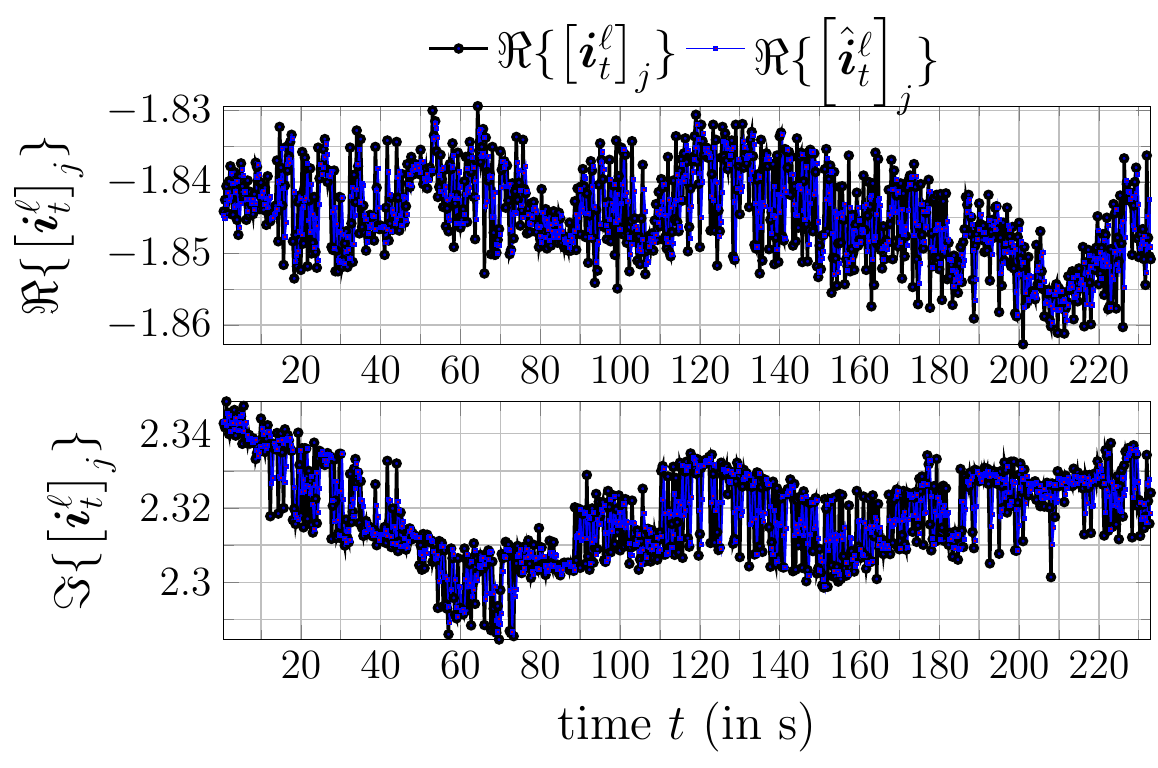}
\caption{AR model for load bus, $j=1312$, i.e. bus with highest absolute value of load.}
\label{fig:load_dynamics_AR}
\end{figure}
To emphasize the temporal nature of the input, the $2$-dimensional frequency response (in both graph and time domains) is plotted for the input $\tilde{\bm{x}}_{t}$ in Fig. \ref{fig:2D_freq}. 
The figure provides evidence of the coupling between the graph frequencies and time series Fourier power spectrum, and the variability of the temporal response depending on what GFT frequency mode is excited with Fourier spectra that are more or less concentrated towards low frequencies depending on the GFT mode. Hence, temporal dynamics can inform about what is happening in space (i.e. the trends are coupled).
\begin{figure}
	\centering
	\includegraphics[width= 0.8\columnwidth]{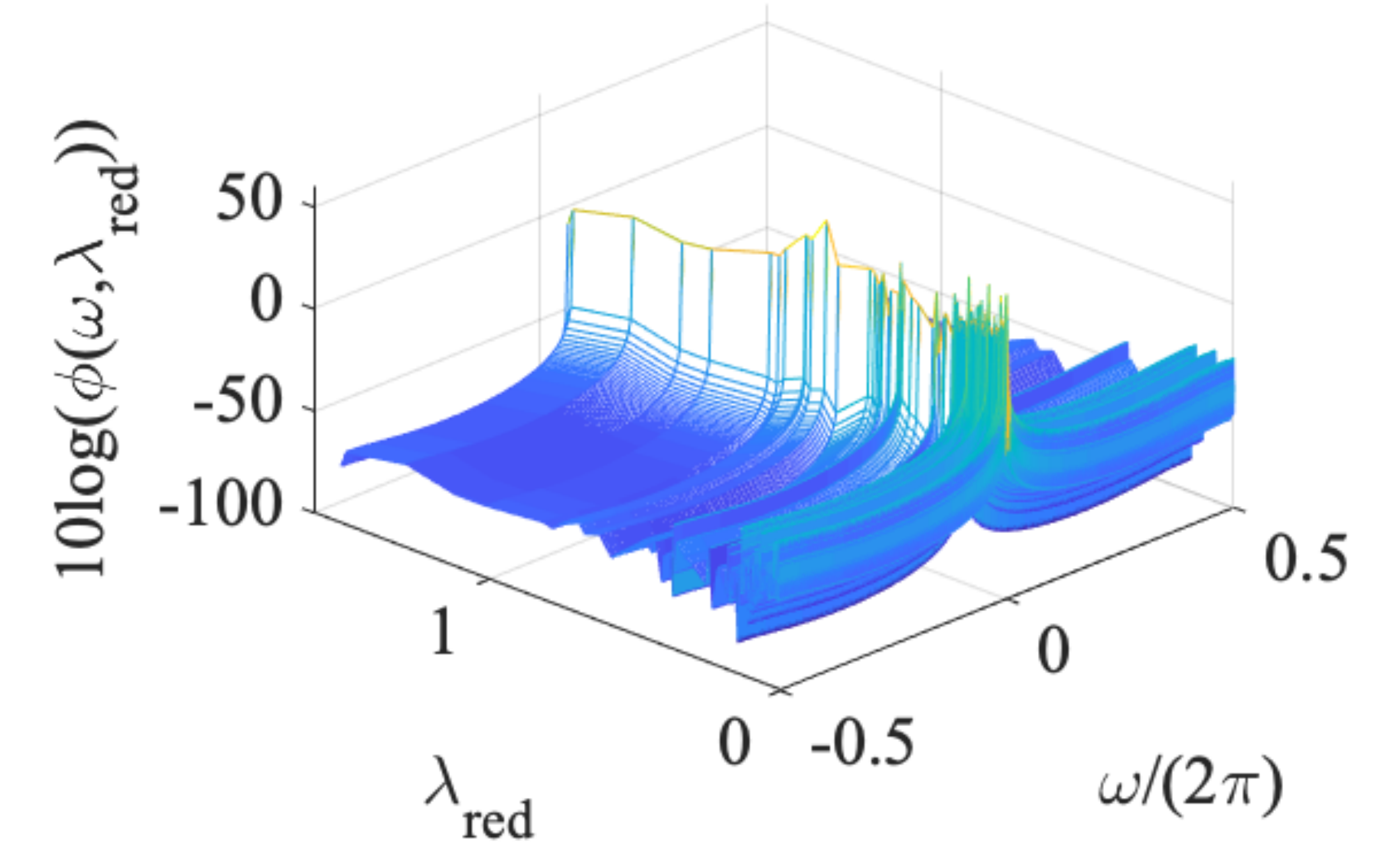}
	\caption{$2-$D frequency response for $\tilde{\bm{x}}_{t}$ }
	\label{fig:2D_freq}
\end{figure}
\begin{figure}
	\centering
	\includegraphics[width= 0.62\columnwidth]{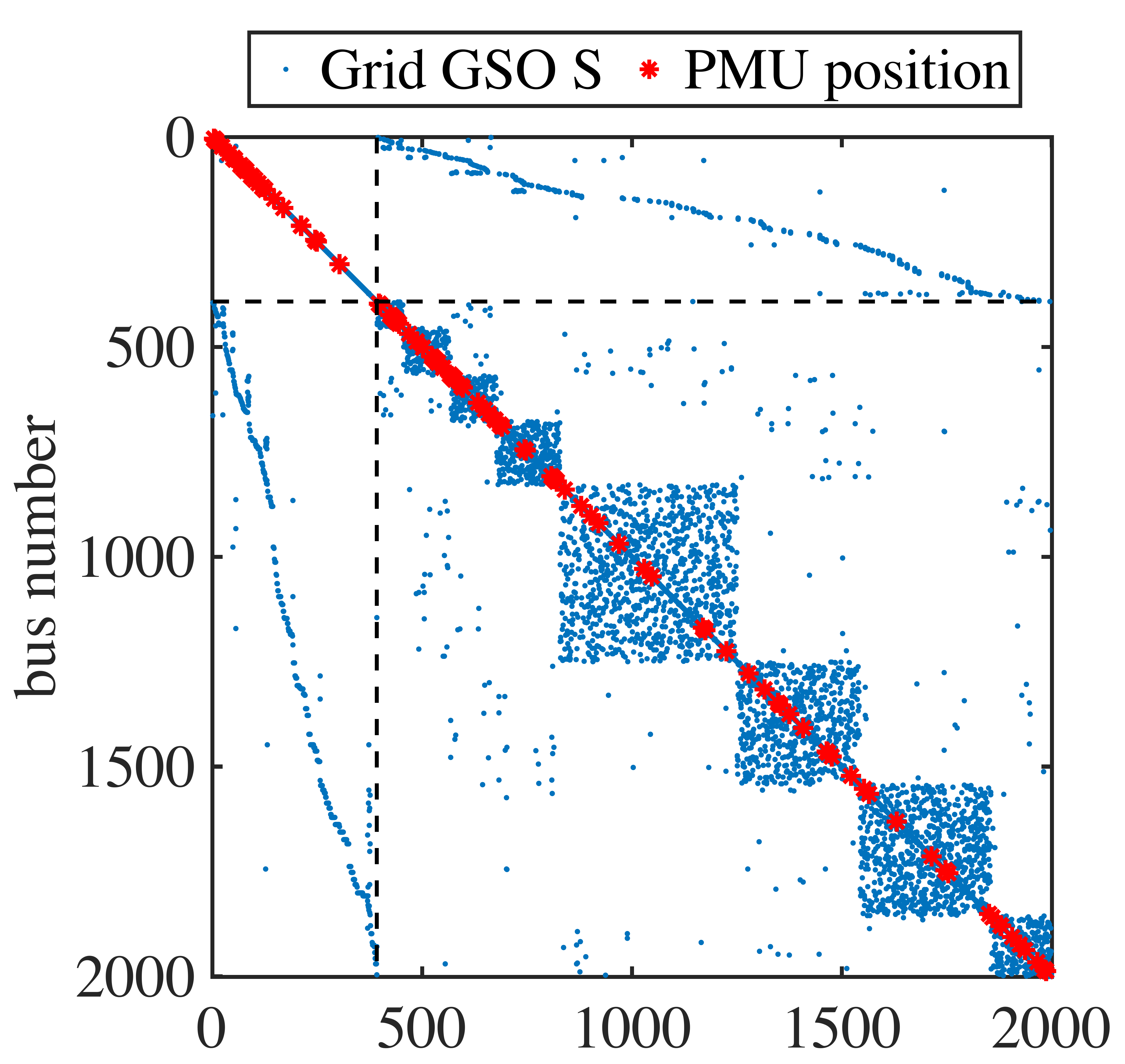}
	\caption{Optimal placement of $\abs{\mathcal{M}} = 100$ PMUs. $\abs{\mathcal{K}}=100$.  }
	\label{fig:PMU_placement_spy}
\end{figure}
\begin{figure}
	\centering
	\includegraphics[width=0.7 \columnwidth]{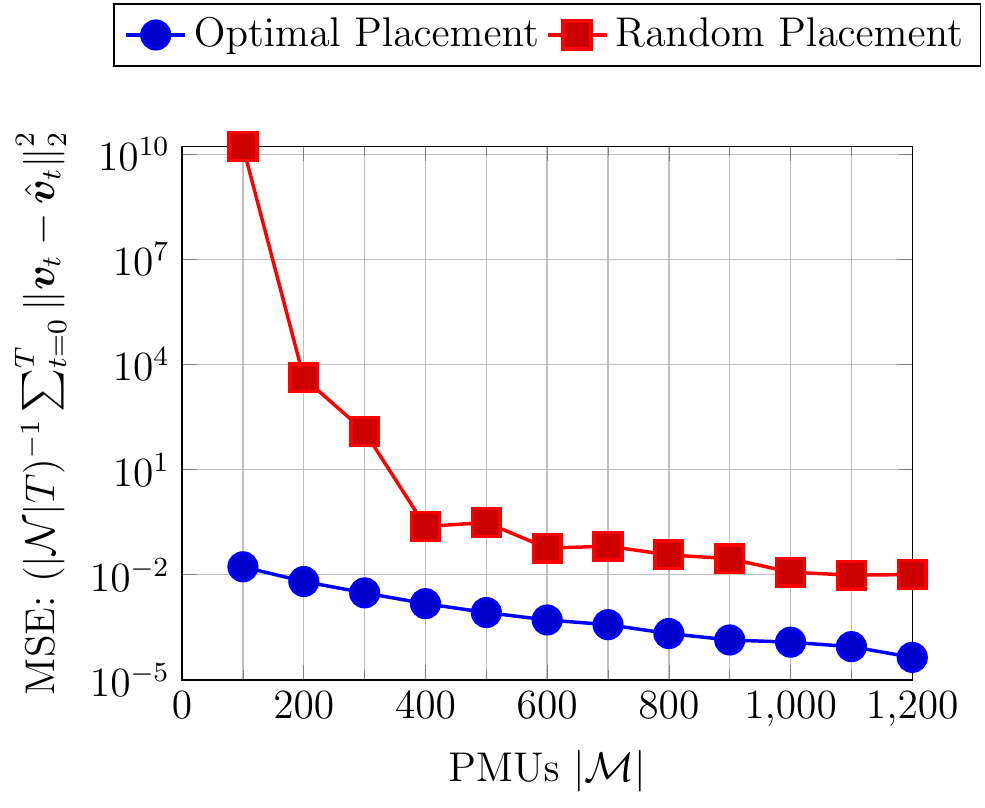}
	\caption{Reconstruction performance  after optimal placement \cite{tsitsvero2016signals} of $\abs{\mathcal{M}}$ PMUs.
	Number of frequencies used: $\abs{\mathcal{K}}=\abs{\mathcal{M}}$. For random placement, $\abs{\mathcal{K}} =100$ used. } 
	\label{fig:PMU_placement_MSE}
\end{figure}
\begin{figure*}
\centering 
\includegraphics[height=0.3\textheight]{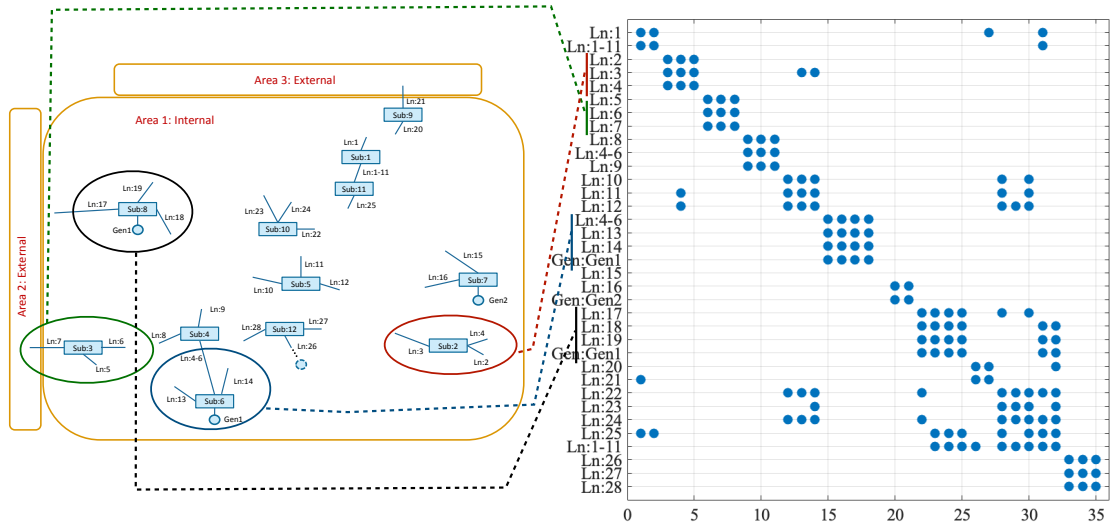}	
\caption{\response{ The map of PMUs placed in ISO-NE test case $3$ \cite{TestCase} (left) and the support of estimated GSO via \eqref{eq:nw_inference} (right) shown. Note that the community structure corresponds to groups of PMUs in the actual system as highlighted in the figure for a few clusters. }}
\label{fig:estimated_GSO}
\end{figure*}
\begin{figure}
	\centering
	\includegraphics[width= \columnwidth]{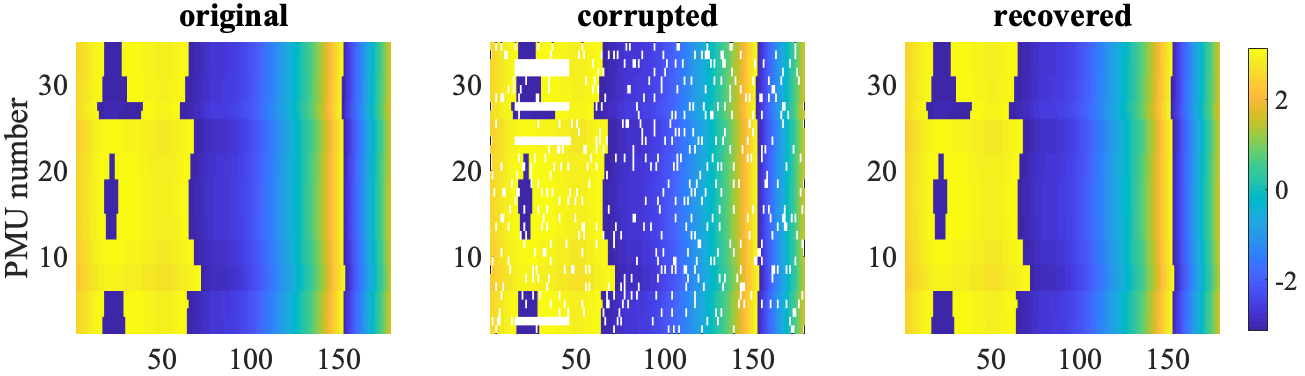}
	\caption{Interpolation of missing measurements for an ISO-NE case using GSO based regularization.
		Note the contiguous missing of samples and our ability to interpolate. The relative noise level used is, $(\abs{\mathcal{M}}T)\sigma^2/\norm{\mathbf{V}}_{F}^{2} = 10^{-4} $ Normalized MSE for this run is $6.22 \times 10^{-4}$.} 
	\label{fig:interpolation_missing_measurements}
\end{figure}
\par\underline{\it Revising GSP tools: sampling and optimal placement}
Fig.~\ref{fig:PMU_placement_spy} shows the placement of $\abs{\mathcal{M}} = 100$ PMUs super-imposed on the support of the ordered Grid-GSO, $\mathbf{S}$ when $\abs{\mathcal{K}}=100$ graph frequency components are considered. Note the distribution of PMUs to different communities as well as on the generator buses as they belong to different graph communities. 
 Fig. \ref{fig:PMU_placement_MSE} exhibits the performance of the GSP based reconstruction method  on optimally placed $\abs{\mathcal{M}}$ PMUs that provide down-sampled measurements, $\bm{v}_{\mathcal{M}}$. The number of graph frequencies considered for reconstruction are $\abs{\mathcal{K}}=\abs{\mathcal{M}}$. Even with just $5\%$ of measurements ($100$ PMUs), the reconstruction error is extremely low. 
 For random placement, $\abs{\mathcal{M}}$ PMUs are chosen at random  and $\abs{\mathcal{K}} =100$ graph frequencies are used for reconstruction. The trial of random placement is repeated $1,000$  times and the most frequently occuring error (estimate of mode of the error distribution) is plotted. As expected, the reconstruction error for random placement is orders of magnitude higher than optimal placement.
 \par To illustrate that the proposed modeling holds and algorithms work well also for real PMU data, in the next numerical experiments we used a real-world dataset of measurements from $35$ PMUs placed in ISO New-England grid (ISO-NE) \cite{TestCase}. 
 The data corresponds to a period of $180$ seconds when a large generator near Ln:2 and Ln:4 introduces oscillations in the system. We decimated in time the PMU signals down to sampling frequency $1$ sample/s.  
 
 \par \underline{\it Network inference}: As the underlying GSO is unknown, it is estimated via \eqref{eq:nw_inference} with the goal of recovering the underlying reduced-GSO. \response{Since admittance values are not given, we only compare the support of the estimated GSO with the community of PMUs in the network.
 Fig. \ref{fig:estimated_GSO} shows the support of the estimated GSO and compares it with the map of PMUs highlighting a few clusters of correspondence. From Fig. \ref{fig:estimated_GSO} we see that the block-diagonal nature of the estimated GSO captures the community structure in the map.}
 
 \par \underline{\it Interpolation of missing measurements}: Once the GSO is estimated, we consider the interpolation problem in \eqref{eq:interp_GSO} for the same ISO-NE dataset.  We delete data at random and add noise. We solve the problem in \eqref{eq:interp_GSO} to recover missing measurements. 
 In Fig.\ref{fig:interpolation_missing_measurements} we compare the original, corrupted and recovered measurements. Corrupted measurements have missing samples not just at random but also contiguous in time. The normalized MSE, $\norm{\mathbf{V}-\hat{\mathbf{V}}}_{F}^{2}/\norm{\mathbf{V}}_{F}^{2}$ is the metric used to gauge the reconstruction performance. 
As a comparison, we tested on the same data the AM-FIHT algorithm proposed in \cite{Hankel}, which regularizes the reconstruction task assuming that the Hankel matrix formed with the columns of $\mathbf{V}$, i.e. ${\mathbb{H}(\mathbf{V})}$,  has low rank $r$,  
\begin{align}
\min_{\mathbf{V}} \norm{\mathbb{P}_{\Omega}\left(\hat{\mathbf{V}}-\mathbf{V}\right)}_{F}^{2} ~~\text{subject to} ~~ \text{rank}\left({\mathbb{H}(\mathbf{V})}\right) = r
\label{eq:AM_FIHT}
\end{align}
\begin{figure}
    \centering
    \includegraphics[width=0.75\columnwidth]{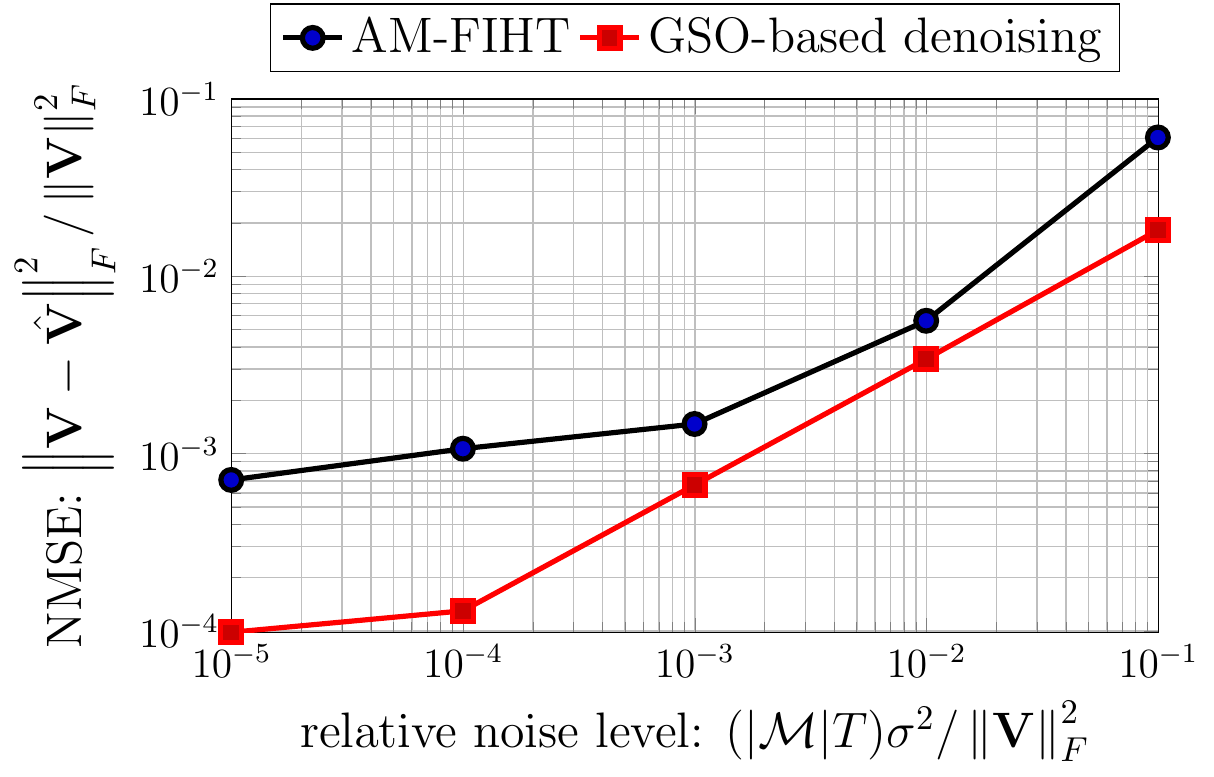}
    \caption{Comparison of AM-FIHT algorithm in \cite{Hankel} ($r=10, n1=3, \beta =0$) and the proposed GSO based interpolation with $50\%$ of missing measurements  in the ISO-NE dataset.}
    \label{fig:AM_FIHT}
\end{figure}
The plot comparing the two methods is shown in Fig. \ref{fig:AM_FIHT}. As seen, the GSO based method outperforms the AM-FIHT for this dataset, indicating that the regularization using the GSO is more effective at capturing the low-rank nature of the data, compared to seeking an arbitrary low rank structure in the the Hankel matrix of the data. 
\begin{figure}
	\centering
	\includegraphics[width=0.68\columnwidth]{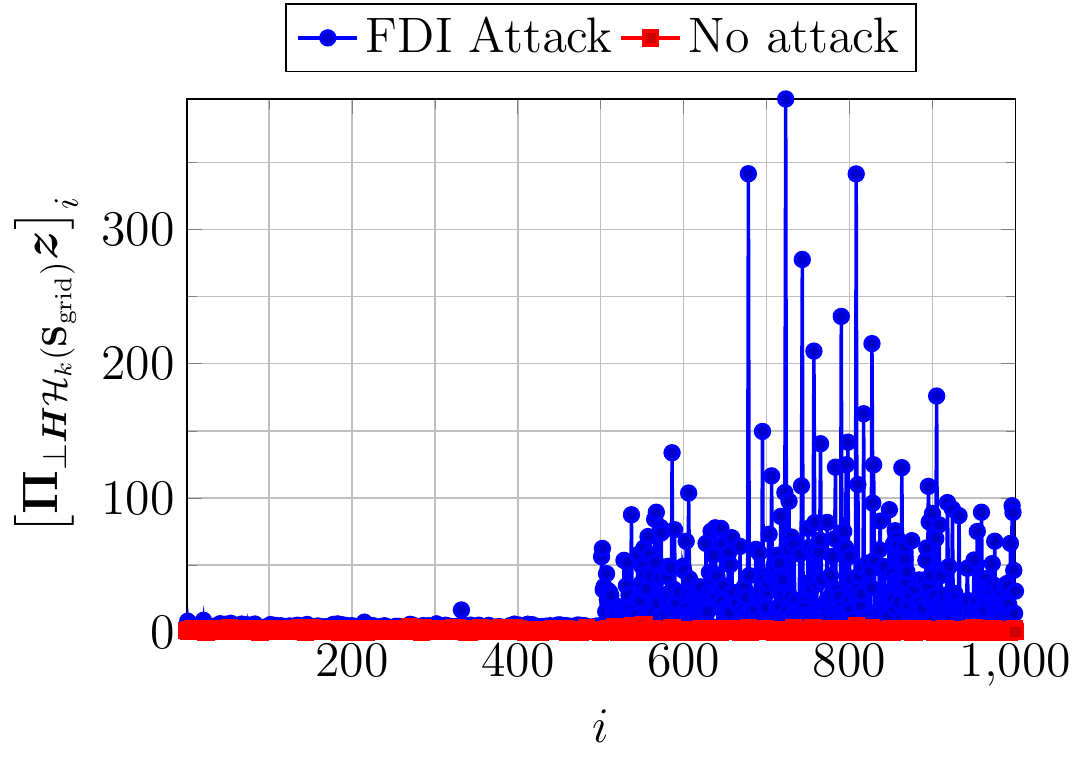}
	\caption{Components of projection of received measurement $\bm{z}$ onto the orthogonal subspace, $\bm{\Pi}_{\perp \bm{H} \mathcal{H}_{k}( \mathbf{S}) } {\bm z}   $, $  \abs{{\mathcal{A}}} = 500,  \abs{\mathcal{K}} = 200,  \abs{{\mathcal{C}}} = 250$.}
	\label{fig:FDI_freq_content}
\end{figure}
\begin{figure}
	\centering
	\includegraphics[width=0.7\columnwidth]{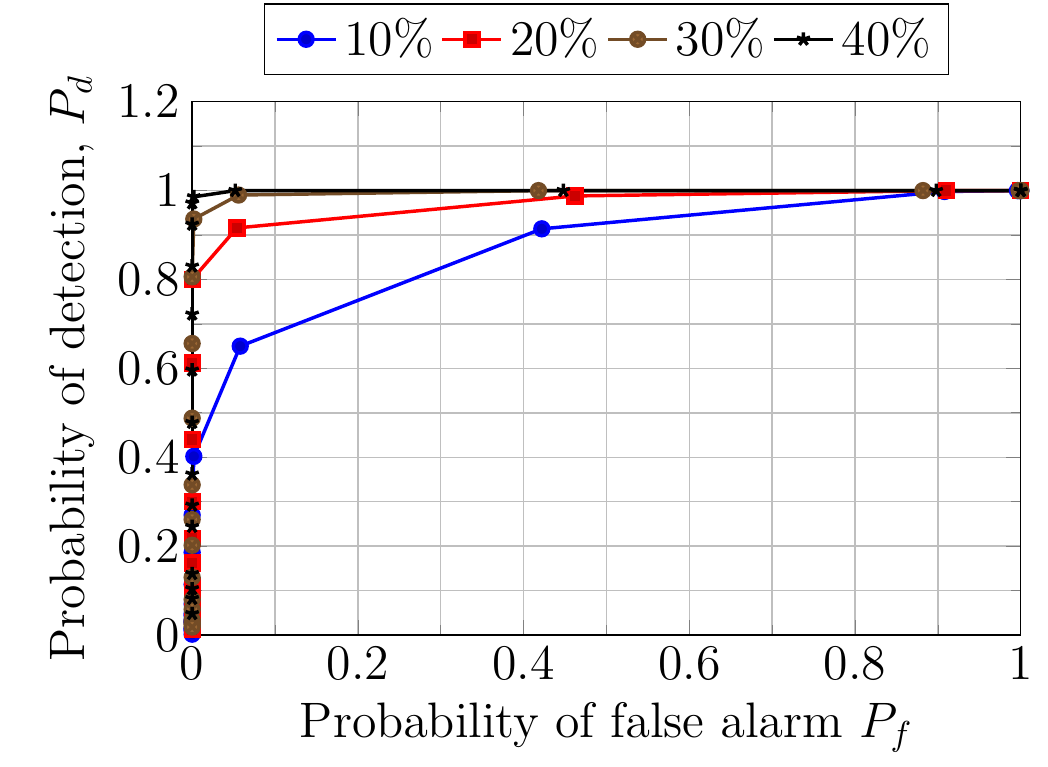} 
	\caption{ Empirical ROC curve for different $\abs{{\mathcal{C}}} $  when a percentage of them are malicious,  $\abs{{\mathcal{C}}}/\abs{{\mathcal{A}}} \times 100 $ with $\abs{{\mathcal{A}}} = 500 $ (out of $2,000$) are available.   }
	\label{fig:ROC}
\end{figure}
	\begin{figure}
		\centering
		\includegraphics[width=0.75\columnwidth]{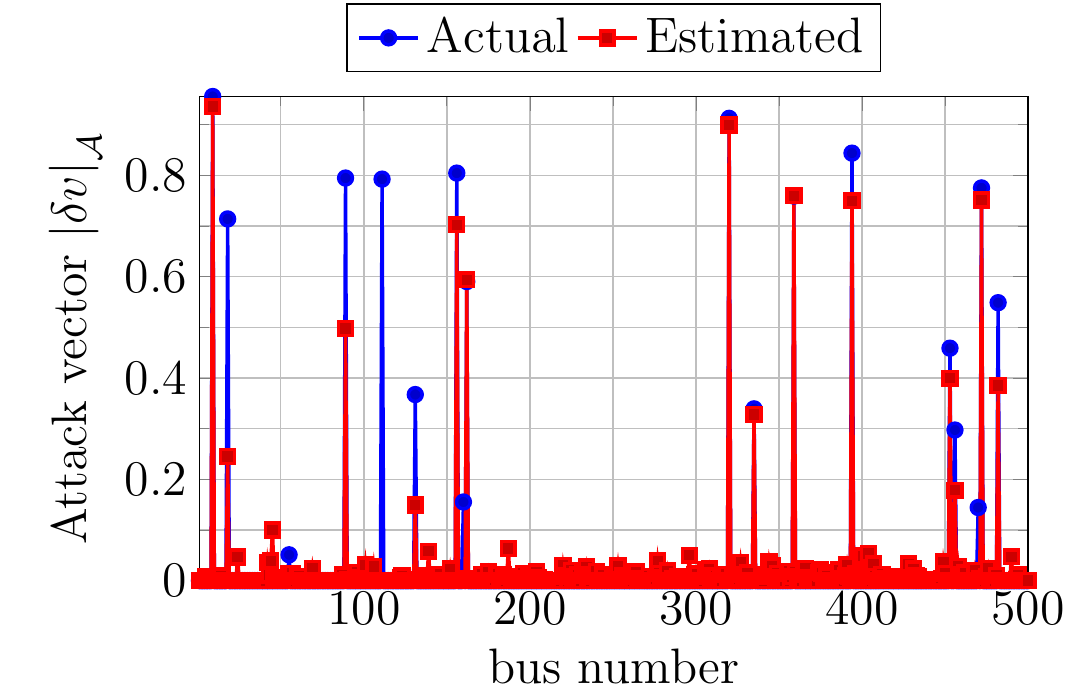} 
		\caption{ Reconstruction of attack vector.   }
		\label{fig:FDI_recons}
	\end{figure}
\begin{figure}
	\centering
	\includegraphics[width=0.8\columnwidth]{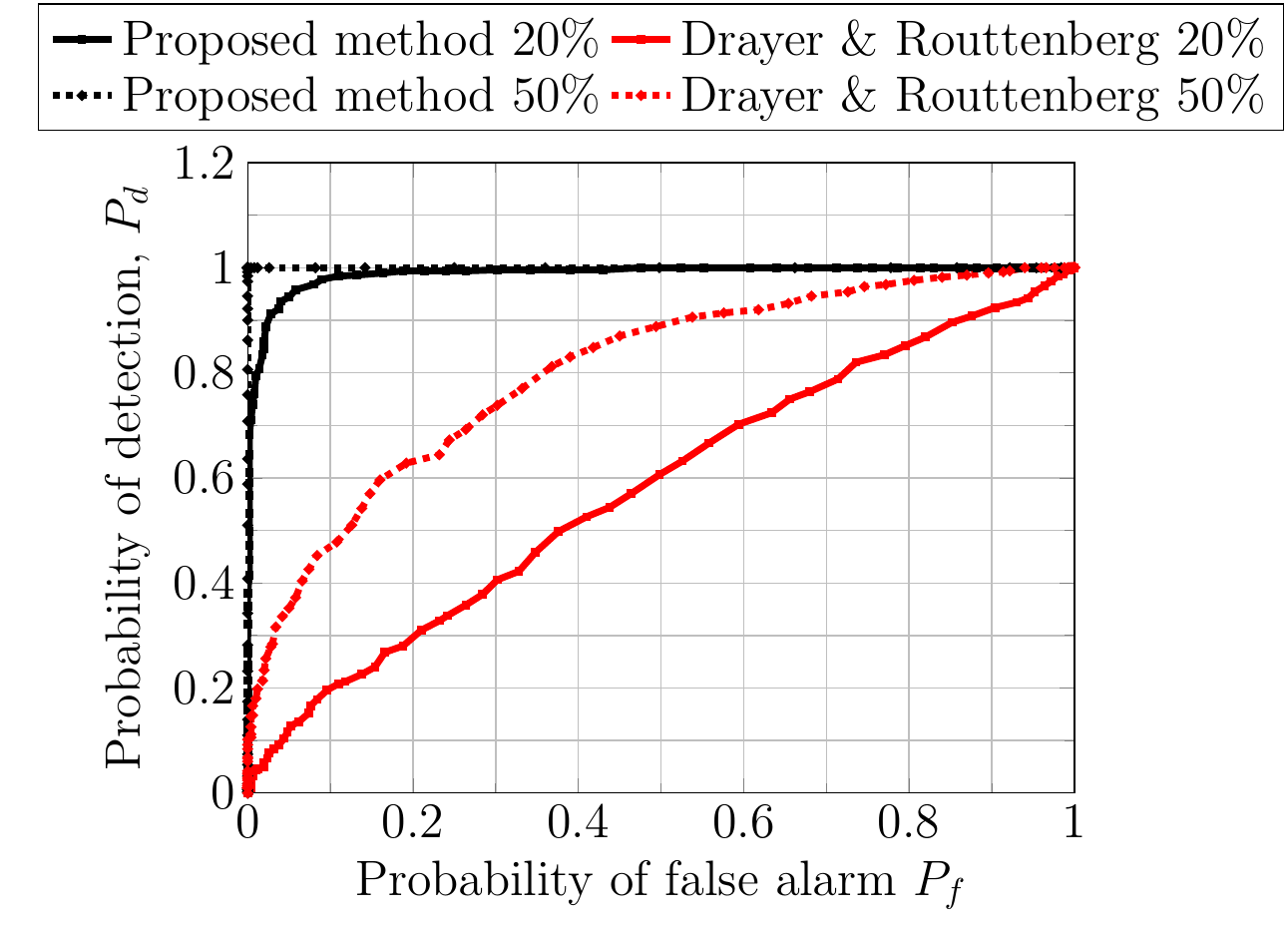} 
	\caption{Empirical ROC curve for methods proposed here and by Drayer \& Routtenberg \cite{Drayer2019} when all voltage measurements are available, $\abs{\mathcal{A}} = 2,000$. A percent of the measurements, $\abs{{\mathcal{C}}}/\abs{{\mathcal{A}}} \!\times\! 100$ are malicious. The relative noise level is $10^{-2}$.  }
	\label{fig:ROC_comparison}
\end{figure}
\par \underline{\it Detection of FDI attacks}: 
Fig.\ref{fig:FDI_freq_content} shows the magnitude of the projection of the received measurement $\bm{z}$ on the orthogonal subspace $\bm{\Pi}_{\perp \bm{H} \mathcal{H}_{k}( \mathbf{S})}$. From Fig. \ref{fig:FDI_freq_content} it is evident that when there is no attack, the magnitude of the projected component is orders of magnitude lower than when the measurements are under the FDI attack. This validates the idea of using high GFT frequency activity as an indicator of anomalies. 
Fig. \ref{fig:ROC} shows the empirical receiver operator characteristics (ROC) curve highlighting the detection performance of the proposed FDI attack detection scheme. The detection performance remains good,  even when very few buses are attacked. 
 We compare the performance of the proposed FDI attack  detection with that of the method in \cite{Drayer2019} when the full state i.e. when all voltage measurements are available, $\abs{\mathcal{A}} = 2,000$. The underlying principle  to detect the attack in \cite{Drayer2019} is to look at the magnitude of graph frequency components at higher graph frequencies which is similar in principle to the detection test we undertake. 
They use the real and imaginary parts of the system admittance matrix as $2$ GSOs, $  \Re \{\bm{Y}\} $ and $  \Im\{\bm{Y}\} $ respectively. Their test statistic is comprised of four components that are the frequency response of high-pass filtered real and imaginary voltage measurements (see Algorithm.2 in \cite{Drayer2019}). 
Fig. \ref{fig:ROC_comparison}  shows the empirical ROC curves that compare the performance of the proposed method and the one in  \cite{Drayer2019} when all voltage measurements are available and are noisy. The relative noise level used is $10^{-2}$.  As evident from the curves, the proposed method performs better than the method in \cite{Drayer2019}. This is because our test statistic is more robust to noise and also significantly more sensitive in detecting the attack vectors, even when only few buses are attacked.

Fig. \ref{fig:FDI_recons} shows the reconstruction of magnitude of the attack vector $\delta \bm{v}$ when $500$ measurements are available and number of attacked buses $\abs{\mathcal{C}} = 50$.		
\par \underline{\it Compression based results}:
\begin{figure}
	\centering
\includegraphics[width=0.7\columnwidth]{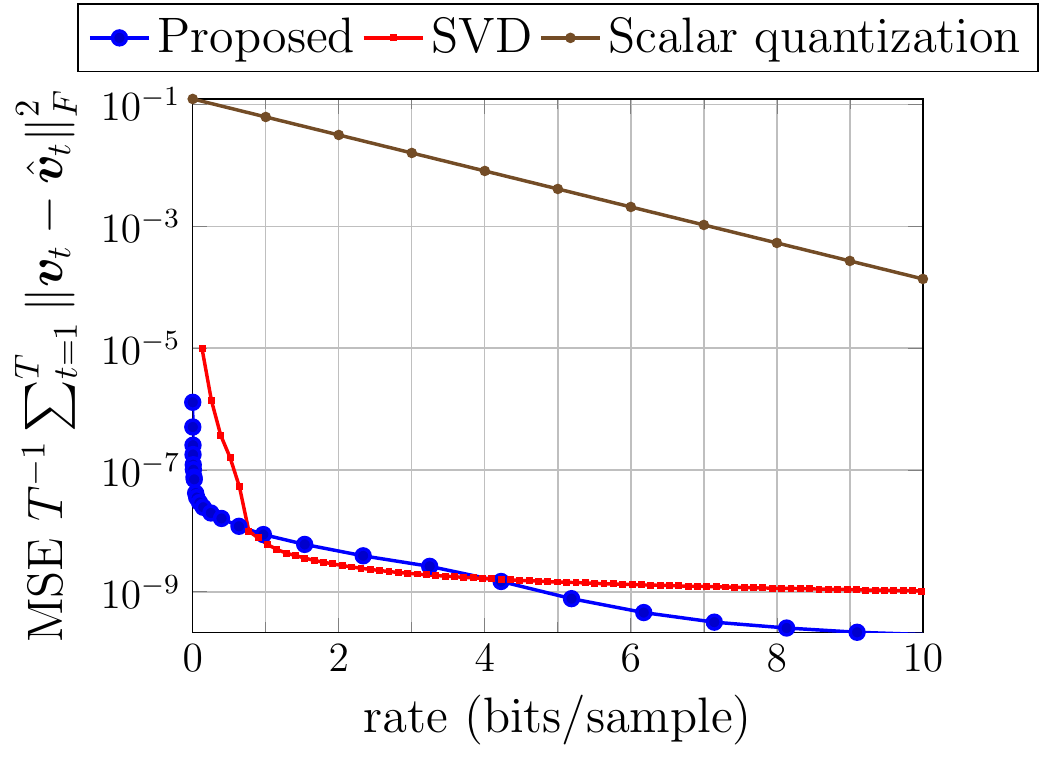}
	\caption{Empirical rate distortion (RD) curve for the proposed compression method compared with singular value thresholding and quantization. }
	\label{fig:RD_curve_full}
\end{figure}
For voltage data compression, we compared with two schemes: scalar quantization and singular value thresholding (SVT) from \cite{DataCompressionSVD}.   Fig.~\ref{fig:RD_curve_full} plots the empirical rate-distortion (RD) curve and shows the comparison between all $3$ schemes. As expected, scalar quantization does poorly compared to the other schemes. 
The SVT scheme simply uses few of the largest singular vectors for data reconstruction. Considering that it is indicative of voltage graph signal lying in a low-dimensional subspace, it is not surprising that the SVT scheme does well. However, the SVT curve rate-distortion curve eventually saturates.  Note that the performance of the proposed method are comparable to those of the SVT. However, the latter is a batch method, while the proposed method is sequential, which has important implication for the online communications of PMU data. 

\section{Conclusions}\label{sec:conclusions}
In this paper, we proposed the framework of Grid-GSP for the power grid that highlights the inherent spatio-temporal structure in the voltage phasors by employing concepts from GSP. 
 Grid-GSP revisits the concepts of sampling and reconstruction,   interpolation, network inference and applications, to detection of FDI  attacks and a lossy sequential data compression, were introduced using the lens of GSP. 
The resulting algorithms were tested on data from both synthetic and real-world datasets. The paper opens the door to leverage the GSP foundations for all types of grid data analytical tasks.  
	\bibliographystyle{IEEEtran}
	\bibliography{./Bibliography/bib_all,./Bibliography/KosutSankar_Papers,./Bibliography/KosutBibs,./Bibliography/Refs_Anamitra_2,./Bibliography/SK_AllRefs,./Bibliography/PMU_References,./Bibliography/MatrixCompletion,./Bibliography/DSW,./Bibliography/PMU,./Bibliography/bib_JSAC,./Bibliography/ref_list_graph}

\begin{thebibliography}{10}
\providecommand{\url}[1]{#1}
\csname url@samestyle\endcsname
\providecommand{\newblock}{\relax}
\providecommand{\bibinfo}[2]{#2}
\providecommand{\BIBentrySTDinterwordspacing}{\spaceskip=0pt\relax}
\providecommand{\BIBentryALTinterwordstretchfactor}{4}
\providecommand{\BIBentryALTinterwordspacing}{\spaceskip=\fontdimen2\font plus
\BIBentryALTinterwordstretchfactor\fontdimen3\font minus
  \fontdimen4\font\relax}
\providecommand{\BIBforeignlanguage}[2]{{%
\expandafter\ifx\csname l@#1\endcsname\relax
\typeout{** WARNING: IEEEtran.bst: No hyphenation pattern has been}%
\typeout{** loaded for the language `#1'. Using the pattern for}%
\typeout{** the default language instead.}%
\else
\language=\csname l@#1\endcsname
\fi
#2}}
\providecommand{\BIBdecl}{\relax}
\BIBdecl

\bibitem{DSW2019}
R.~Ramakrishna and A.~Scaglione, ``{On Modeling Voltage Phasor Measurements as
  Graph Signals},'' in \emph{{IEEE Data Science Workshop (DSW) 2019}}, June
  2019, pp. {275--279}.

\bibitem{GlobalSIP2019}
R.~Ramakrishna and A.~Scaglione, ``{Detection of False Data Injection Attack
  using Graph Signal Processing for the Power Grid},'' in \emph{{2019 IEEE
  Global Conference on Signal and Information Processing (GlobalSIP). IEEE}},
  2019.

\bibitem{gridGSP_ortega}
D.~I. Shuman, S.~K. Narang, P.~Frossard, A.~Ortega, and P.~Vandergheynst, ``The
  emerging field of signal processing on graphs: Extending high-dimensional
  data analysis to networks and other irregular domains,'' \emph{{IEEE Signal
  Processing Magazine}}, 2013.

\bibitem{GFDesign}
N.~Tremblay, P.~Gon{\c{c}}alves, and P.~Borgnat, ``Design of graph filters and
  filterbanks,'' in \emph{Cooperative and Graph Signal Processing}.\hskip 1em
  plus 0.5em minus 0.4em\relax Elsevier, 2018, pp. 299--324.

\bibitem{GSP_Moura}
A.~Sandryhaila and J.~M.~F. Moura, ``{Discrete Signal Processing on Graphs:
  Frequency Analysis},'' \emph{{IEEE Transactions on Signal Processing}},
  vol.~62, no.~12, June 2014.

\bibitem{phadke2006history}
A.~Phadke and J.~Thorp, ``{History and Applications of Phasor Measurements},''
  in \emph{{2006 IEEE PES Power Systems Conference and Exposition}}.\hskip 1em
  plus 0.5em minus 0.4em\relax IEEE, 2006, pp. 331--335.

\bibitem{GraphPowerSystem}
T.~Ishizaki, A.~Chakrabortty, and J.-I. Imura, ``{Graph-Theoretic Analysis of
  Power Systems},'' \emph{{Proceedings of the IEEE}}, vol. 106, no.~5, pp.
  931--952, May 2018.

\bibitem{Dorfler2013}
F.~D{\"o}rfler and F.~Bullo, ``{Kron Reduction of Graphs With Applications to
  Electrical Networks},'' \emph{{IEEE Transactions on Circuits and Systems-I:
  Regular Papers}}, vol.~60, no.~1, 2013.

\bibitem{kron_reduction_power_network}
F.~Dorfler and F.~Bullo, ``{Spectral Analysis of Synchronization in a Lossless
  Structure-Preserving Power Network Model},'' in \emph{{2010 First IEEE
  International Conference on Smart Grid Communications}}, 2010.

\bibitem{Xie2014}
L.~Xie, Y.~Chen, and P.~R. Kumar, ``Dimensionality reduction of synchrophasor
  data for early event detection: Linearized analysis,'' \emph{IEEE Trans.~on
  Power Systems}, vol.~29, no.~6, pp. 2784--2794, Nov 2014.

\bibitem{LowrankMissingData}
P.~Gao, M.~Wang, S.~G. Ghiocel, J.~H. Chow, B.~Fardanesh, and G.~Stefopoulos,
  ``{Missing Data Recovery by Exploiting Low-dimensionality in Power System
  Synchrophasor Measurements},'' \emph{{IEEE Transactions on Power Systems}},
  vol.~31, no.~2, pp. 1006--1013, 2016.

\bibitem{Wang2017}
M.~Wang, ``{Data quality management of synchrophasor data in power systems by
  exploiting low-dimensional models},'' in \emph{{ 2017 51st Annual Conference
  on Information Sciences and Systems (CISS)}}, 2017.

\bibitem{AnomalyDetection}
M.~Jamei, A.~Scaglione, C.~Roberts, E.~Stewart, S.~Peisert, C.~McParland, and
  A.~McEachern, ``{Anomaly Detection Using Optimally-Placed PMU Sensors in
  Distribution Grids},'' \emph{{IEEE Transactions on Power Systems}}, vol.~33,
  no.~4, pp. 3611--3622, July 2018.

\bibitem{RealTimeIdentification}
W.~Li, M.~Wang, and J.~H. Chow, ``{Real-Time Event Identification Through
  Low-Dimensional Subspace Characterization of High-Dimensional Synchrophasor
  Data},'' \emph{{IEEE Transactions on Power Systems}}, vol.~33, no.~5, Sept
  2018.

\bibitem{Kim2015}
J.~Kim, L.~Tong, and R.~J. Thomas, ``{Subspace Methods for Data Attack on State
  Estimation: A Data Driven Approach},'' \emph{{IEEE Transactions on Signal
  Processing}}, vol.~63, no.~5, pp. 1102--1114, March 2015.

\bibitem{PMUplacement2010}
P.~Du, Z.~Huang, R.~Diao, B.~Lee, and K.~K. Anderson, ``{PMU Placement for
  enhancing dynamic observability of a power grid},'' in \emph{{2010 IEEE
  Conference on Innovative Technologies for an Efficient and Reliable
  Electricity Supply}}, Sept 2010, pp. 15--21.

\bibitem{Pal2014517}
A.~Pal, G.~A. Sanchez-Ayala, V.~A. Centeno, and J.~S. Thorp, ``{A PMU Placement
  Scheme Ensuring Real-Time Monitoring of Critical Buses of the Network},''
  \emph{IEEE Transactions on Power Delivery}, vol.~29, no.~2, pp. 510--517,
  April 2014.

\bibitem{Wang2010}
Z.~Wang, A.~Scaglione, and R.~Thomas, ``{Generating Statistically Correct
  Random Topologies for Testing Smart Grid Communication and Control
  Networks},'' \emph{{IEEE Transactions on Power Systems}}, vol.~1, no.~1, pp.
  28--39, 2010.

\bibitem{li2013blind}
{Li, Xiao and Poor, H Vincent and Scaglione, Anna}, ``{Blind topology
  identification for power systems},'' in \emph{{2013 IEEE International
  Conference on Smart Grid Communications (SmartGridComm)}}.\hskip 1em plus
  0.5em minus 0.4em\relax {IEEE}, 2013, pp. {91--96}.

\bibitem{Grotas2019}
S.~{Grotas}, Y.~{Yakoby}, I.~{Gera}, and T.~{Routtenberg}, ``{Power Systems
  Topology and State Estimation by Graph Blind Source Separation},''
  \emph{{IEEE Transactions on Signal Processing}}, vol.~67, no.~8, pp.
  2036--2051, 2019.

\bibitem{Xiang2019}
Z.~{Xiang}, K.~{Huang}, W.~{Deng}, and C.~{Yang}, ``{Blind Topology
  Identification for Smart Grid Based on Dictionary Learning},'' in \emph{{2019
  IEEE Symposium Series on Computational Intelligence (SSCI)}}, 2019, pp.
  1319--1326.

\bibitem{Deka2020}
D.~{Deka}, M.~{Chertkov}, and S.~{Backhaus}, ``{Joint Estimation of Topology
  and Injection Statistics in Distribution Grids With Missing Nodes},''
  \emph{{IEEE Transactions on Control of Network Systems}}, vol.~7, no.~3, pp.
  1391--1403, 2020.

\bibitem{talukdar2020physics}
S.~Talukdar, D.~Deka, H.~Doddi, D.~Materassi, M.~Chertkov, and M.~V. Salapaka,
  ``{Physics informed topology learning in networks of linear dynamical
  systems},'' \emph{{Automatica}}, vol. 112, 2020.

\bibitem{GlobalSIP2018}
E.~Drayer and T.~Routtenberg, ``{Detection of False Data Injection Attacks in
  Power Systems with Graph Fourier Transform},'' in \emph{2018 IEEE Global
  Conference on Signal and Information Processing (GlobalSIP). IEEE, 2018},
  2018, pp. {890--894}.

\bibitem{Drayer2019}
E.~Drayer and T.~Routtenberg, ``{Detection of False Data Injection Attacks in
  Smart Grids Based on Graph Signal Processing},'' \emph{{IEEE Systems
  Journal}}, August 2019.

\bibitem{JSAC}
M.~Jamei, R.~Ramakrishna, T.~Tesfay, R.~Gentz, C.~Roberts, A.~Scaglione, and
  S.~Peisert, ``{Phasor Measurement Units Optimal Placement and Performance
  Limits for Fault Localization},'' \emph{{ IEEE Journal on Selected Areas in
  Communications }}, 2019.

\bibitem{InterAreaOsc}
L.~Fan, ``{Interarea Oscillations Revisited},'' \emph{{IEEE Transactions on
  Power Systems}}, vol.~32, no.~2, pp. {1585--1586}, 2017.

\bibitem{LocalOscHuang}
T.~Huang, N.~M. Freris, P.~R. Kumar, and L.~Xie, ``{Localization of Forced
  Oscillations in the Power Grid Under Resonance Conditions},'' in \emph{{2018
  52nd Annual Conference on Information Sciences and Systems (CISS)}}, 2018.

\bibitem{GraphLaplacian}
L.~Guo, C.~Zhao, and S.~H. Low, ``{Graph Laplacian Spectrum and Primary
  Frequency Regulation},'' in \emph{{2018 IEEE Conference on Decision and
  Control (CDC)}}, 2018, pp. {158--165}.

\bibitem{weng2013graphical}
Y.~Weng, R.~Negi, and M.~D. Ili{\'c}, ``{Graphical model for state estimation
  in electric power systems},'' in \emph{{2013 IEEE International Conference on
  Smart Grid Communications (SmartGridComm)}}.\hskip 1em plus 0.5em minus
  0.4em\relax {IEEE}, 2013, pp. {103--108}.

\bibitem{deka2019topology}
D.~Deka, M.~Chertkov, and S.~Backhaus, ``{Topology Estimation Using Graphical
  Models in Multi-Phase Power Distribution Grids},'' \emph{{IEEE Transactions
  on Power Systems}}, vol.~35, no.~3, pp. {1663--1673}, {May} 2020.

\bibitem{dvijotham2017graphical}
K.~Dvijotham, M.~Chertkov, P.~Van~Hentenryck, M.~Vuffray, and S.~Misra,
  ``{Graphical models for optimal power flow},'' \emph{{Constraints}}, vol.~22,
  no.~1, pp. {24--49}, 2017.

\bibitem{bobba2010detecting}
R.~B. Bobba, K.~M. Rogers, Q.~Wang, H.~Khurana, K.~Nahrstedt, and T.~J.
  Overbye, ``Detecting false data injection attacks on {DC} state estimation,''
  in \emph{Preprints of the First Workshop on Secure Control Systems, CPSWEEK},
  vol. 2010, 2010.

\bibitem{dan2010stealth}
G.~Dan and H.~Sandberg, ``Stealth attacks and protection schemes for state
  estimators in power systems,'' in \emph{Smart Grid Communications
  (SmartGridComm), 2010 First IEEE International Conference on}, 2010, pp.
  214--219.

\bibitem{kosut2011malicious}
O.~Kosut, L.~Jia, R.~J. Thomas, and L.~Tong, ``Malicious data attacks on the
  smart grid,'' \emph{IEEE Transactions on Smart Grid}, vol.~2, no.~4, pp.
  645--658, 2011.

\bibitem{liang2017review}
G.~Liang, J.~Zhao, F.~Luo, S.~R. Weller, and Z.~Y. Dong, ``A review of false
  data injection attacks against modern power systems,'' \emph{IEEE
  Transactions on Smart Grid}, vol.~8, no.~4, pp. 1630--1638, July 2017.

\bibitem{Esnaola2016}
I.~Esnaola, S.~M. Perlaza, H.~V. Poor, and O.~Kosut, ``{Maximum Distortion
  Attacks in Electricity Grids},'' \emph{{IEEE Transactions on Smart Grid}},
  vol.~7, no.~4, pp. 2007--2015, Jul. 2016.

\bibitem{Zhang2017}
J.~Zhang, Z.~Chu, L.~Sankar, and O.~Kosut, ``False data injection attacks on
  phasor measurements that bypass low-rank decomposition,'' in \emph{IEEE
  International Conference on Smart Grid Communications (SmartGridComm)}, Oct.
  2017.

\bibitem{he2017real}
Y.~He, G.~J. Mendis, and J.~Wei, ``Real-time detection of false data injection
  attacks in smart grid: A deep learning-based intelligent mechanism,''
  \emph{IEEE Transactions on Smart Grid}, vol.~8, no.~5, pp. 2505--2516, Sept
  2017.

\bibitem{shepard2012evaluation}
D.~P. Shepard, T.~E. Humphreys, and A.~A. Fansler, ``Evaluation of the
  vulnerability of phasor measurement units to gps spoofing attacks,''
  \emph{International Journal of Critical Infrastructure Protection}, vol.~5,
  no. 3-4, pp. 146--153, 2012.

\bibitem{heng2014reliable}
L.~Heng, J.~J. Makela, A.~D. Dominguez-Garcia, R.~B. Bobba, W.~H. Sanders, and
  G.~X. Gao, ``Reliable {GPS}-based timing for power systems: A multi-layered
  multi-receiver architecture,'' in \emph{Power and Energy Conference at
  Illinois (PECI), 2014}, 2014, pp. 1--7.

\bibitem{Gao2016_PMUIdentification}
P.~Gao, M.~Wang, J.~H. Chow, S.~G. Ghiocel, B.~Fardanesh, G.~Stefopoulos, and
  M.~P. Razanousky, ``Identification of successive {"}unobservable {"} cyber
  data attacks in power systems through matrix decomposition,'' \emph{IEEE
  Transactions on Signal Processing}, vol.~64, no.~21, pp. 5557--5570, Nov
  2016.

\bibitem{Top2013}
P.~Top and J.~Breneman, ``Compressing phasor measurement data,'' in \emph{2013
  IEEE Power Energy Society General Meeting}, July 2013, pp. 1--4.

\bibitem{Klump2010}
R.~Klump, P.~Agarwal, J.~E. Tate, and H.~Khurana, ``Lossless compression of
  synchronized phasor measurements,'' in \emph{IEEE PES General Meeting}, July
  2010, pp. 1--7.

\bibitem{Tate2016}
J.~E. Tate, ``Preprocessing and {Golomb-Rice} encoding for lossless compression
  of phasor angle data,'' \emph{IEEE Trans.~on Smart Grid}, vol.~7, no.~2, pp.
  718--729, March 2016.

\bibitem{KirtiHICCS}
S.~Kirti, Z.~Wang, A.~Scaglione, and R.~Thomas, ``{On the Communication
  Architecture for Wide-Area Real-Time Monitoring in Power Networks},'' in
  \emph{{2007 40th Annual Hawaii International Conference on System Sciences
  (HICSS'07)}}, 2007.

\bibitem{Gadde2016}
P.~H. Gadde, M.~Biswal, S.~Brahma, and H.~Cao, ``Efficient compression of {PMU}
  data in {WAMS},'' \emph{IEEE Trans.~on Smart Grid}, vol.~7, no.~5, pp.
  2406--2413, Sept 2016.

\bibitem{Ge2015}
Y.~Ge, A.~J. Flueck, D.~K. Kim, J.~B. Ahn, J.~D. Lee, and D.~Y. Kwon, ``Power
  system real-time event detection and associated data archival reduction based
  on synchrophasors,'' \emph{IEEE Trans.~on Smart Grid}, vol.~6, no.~4, pp.
  2088--2097, July 2015.

\bibitem{DataCompressionSVD}
J.~C.~S. de~Souza, T.~M.~L. Assis, and B.~C. Pal, ``{Data Compression in Smart
  Distribution Systems via Singular Value Decomposition},'' \emph{{IEEE
  Transactions on Smart Grid}}, vol.~8, no.~1, January 2017.

\bibitem{Mehra2013}
R.~Mehra, V.~Patel, F.~Kazi, N.~M. Singh, and S.~R. Wagh, ``Modes preserving
  wavelet based multi-scale {PCA} algorithm for compression of smart grid
  data,'' in \emph{2013 International Conference on Advances in Computing,
  Communications and Informatics (ICACCI)}, Aug 2013, pp. 817--821.

\bibitem{isufi2017autoregressive}
E.~Isufi, A.~Loukas, A.~Simonetto, and G.~Leus, ``{Autoregressive Moving
  Average Graph Filtering},'' \emph{{IEEE Transactions on Signal Processing}},
  vol.~65, no.~2, pp. 274--288, 2017.

\bibitem{Ramakrishna2020}
R.~{Ramakrishna}, H.~T. {Wai}, and A.~{Scaglione}, ``{A User Guide to Low-Pass
  Graph Signal Processing and Its Applications: Tools and Applications},''
  \emph{{IEEE Signal Processing Magazine}}, vol.~37, no.~6, pp. 74--85,
  November 2020.

\bibitem{sandryhaila2013discrete}
A.~Sandryhaila and J.~M. Moura, ``{Discrete Signal Processing on Graphs},''
  \emph{IEEE Transactions on Signal Processing}, vol.~61, no.~7, pp.
  1644--1656, 2013.

\bibitem{Matrix_analysis}
R.~A. Horn and C.~R. Johnson, \emph{Matrix Analysis}.\hskip 1em plus 0.5em
  minus 0.4em\relax Cambridge University Press, 1990.

\bibitem{mallat2008wavelet}
S.~Mallat, \emph{{A Wavelet Tour of Signal Processing: The Sparse Way}}.\hskip
  1em plus 0.5em minus 0.4em\relax Academic press, 2009.

\bibitem{Directed_graph_laplacian}
R.~Singh, A.~Chakraborty, and B.~Manoj, ``{Graph Fourier transform based on
  direcetd Laplacian},'' in \emph{{2016 International Conference on Signal
  Processing and Communications (SPCOM)}}, 2016.

\bibitem{isufibanelli2016}
E.~{Isufi}, G.~{Leus}, and P.~{Banelli}, ``{2-Dimensional finite impulse
  response graph-temporal filters},'' in \emph{{2016 IEEE Global Conference on
  Signal and Information Processing (GlobalSIP)}}, 2016.

\bibitem{isufi2017filtering}
E.~Isufi, A.~Loukas, A.~Simonetto, and G.~Leus, ``{Filtering Random Graph
  Processes Over Random Time-Varying Graphs},'' \emph{{IEEE Transactions on
  Signal Processing}}, vol.~65, no.~16, August 2017.

\bibitem{JFT}
F.~Grassi, A.~Loukas, N.~Perraudin, and B.~Ricaud, ``{A Time-Vertex Signal
  Processing Framework: Scalable Processing and Meaningful Representations for
  Time-Series on Graphs},'' \emph{{IEEE Transactions on Signal Processing}},
  vol.~66, no.~3, pp. 817--829, February 2018.

\bibitem{isufi2016separable}
E.~Isufi, A.~Loukas, A.~Simonetto, and G.~Leus, ``{Separable autoregressive
  moving average graph-temporal filters},'' in \emph{{2016 24th European Signal
  Processing Conference (EUSIPCO)}}.\hskip 1em plus 0.5em minus 0.4em\relax
  IEEE, 2016, pp. 200--204.

\bibitem{GloverSarmaOverbye:Book}
J.~D. Glover, M.~S. Sarma, and T.~J. Overbye, \emph{{Power system analysis and
  design}}.\hskip 1em plus 0.5em minus 0.4em\relax Cengage Learning, 2008.

\bibitem{decell1965application}
H.~P. Decell, Jr, ``An application of the cayley-hamilton theorem to
  generalized matrix inversion,'' \emph{{SIAM Review}}, vol.~7, no.~4, pp.
  526--528, 1965.

\bibitem{sato1963techniques}
N.~{Sato} and W.~F. {Tinney}, ``{Techniques for Exploiting the Sparsity or the
  Network Admittance Matrix},'' \emph{{IEEE Transactions on Power Apparatus and
  Systems}}, vol.~82, no.~69, pp. 944--950, 1963.

\bibitem{ZhuFreqResp}
P.~Huynh, H.~Zhu, Q.~Chen, and A.~E.Elbanna, ``{Data-Driven Estimation of
  Frequency Response From Ambient Synchrophasor Measurements},'' \emph{{IEEE
  Transactions on Power Systems}}, vol.~33, no.~6, 2018.

\bibitem{Paganini2020}
F.~{Paganini} and E.~{Mallada}, ``{Global Analysis of Synchronization
  Performance for Power Systems: Bridging the Theory-Practice Gap},''
  \emph{{IEEE Transactions on Automatic Control}}, vol.~65, no.~7, pp.
  3007--3022, 2020.

\bibitem{gao2012dynamic}
F.~Gao, J.~S. Thorp, A.~Pal, and S.~Gao, ``{Dynamic state prediction based on
  Auto-Regressive (AR) Model using PMU data},'' in \emph{{2012 IEEE Power and
  Energy Conference at Illinois}}.\hskip 1em plus 0.5em minus 0.4em\relax IEEE,
  2012, pp. 1--5.

\bibitem{Golub&VanLoan:book}
G.~Golub and C.~V. Loan, \emph{{Matrix Computations}}.\hskip 1em plus 0.5em
  minus 0.4em\relax The Johns Hopkins University Press, 1990.

\bibitem{Wang2015}
M.~Wang, J.~H. Chow, P.~Gao, X.~T. Jiang, Y.~Xia, S.~G. Ghiocel, B.~Fardanesh,
  G.~Stefopolous, Y.~Kokai, N.~Saito, and M.~Razanousky, ``A low-rank matrix
  approach for the analysis of large amounts of power system synchrophasor
  data,'' in \emph{2015 48th Hawaii International Conference on System
  Sciences}, Jan 2015, pp. 2637--2644.

\bibitem{Dahal2012}
N.~Dahal, R.~L. King, and V.~Madani, ``Online dimension reduction of
  synchrophasor data,'' in \emph{{IEEE PES Transmission and Distribution
  Conference and Exposition (T \& D)}}, 2012.

\bibitem{ADMM_PMU}
M.~Liao, D.~Shi, Z.~Yu, Z.~Yi, Z.~Wang, and Y.~Xiang, ``{An Alternating
  Direction Method of Multipliers Based Approach for PMU Data Recovery},''
  \emph{{IEEE Transactions on Smart Grid}}, 2018.

\bibitem{tsitsvero2016signals}
M.~Tsitsvero, S.~Barbarossa, and P.~Di~Lorenzo, ``{Signals on Graphs:
  Uncertainty Principle and Sampling},'' \emph{{IEEE Transactions on Signal
  Processing}}, vol.~64, no.~18, pp. 4845--4860, 2016.

\bibitem{Anis2016}
A.~Anis, A.~Gadde, and A.~Ortega, ``{Efficient Sampling Set Selection for
  Bandlimited Graph Signals Using Graph Spectral Proxies},'' \emph{{IEEE
  Transactions on Signal Processing}}.

\bibitem{HagenKahn1992}
L.~Hagen and A.~B. Kahng, ``{New Spectral Methods for Ratio Cut Partitioning
  and Clustering},'' \emph{{IEEE Transactions on Computer-Aided Design of
  Integrated Circuits and Systems}}, vol.~11, no.~9, pp. {1074--1085}, 1992.

\bibitem{spectral_clustering_tutorial}
U.~Von~Luxburg, ``A tutorial on spectral clustering,'' \emph{Statistics and
  computing, Springer}, vol.~17, no.~4, pp. 395--416, 2007.

\bibitem{dong2019learning}
X.~Dong, D.~Thanou, M.~Rabbat, and P.~Frossard, ``Learning graphs from data: A
  signal representation perspective,'' \emph{IEEE Signal Processing Magazine},
  vol.~36, no.~3, pp. 44--63, 2019.

\bibitem{diff_enc}
W.~Weber, ``{Differential Encoding for Multiple Amplitude and Phase Shift
  Keying Systems},'' \emph{{IEEE Transactions on Communication}}, vol.~26,
  no.~3, 1978.

\bibitem{CoverThomas:Book}
T.~M. Cover and J.~Thomas, \emph{{Elements of Information Theory}}.\hskip 1em
  plus 0.5em minus 0.4em\relax John Wiley, 1991.

\bibitem{BirchfieldACTIVSg}
A.~B. Birchfield, T.~Xu, K.~M. Gegner, K.~S. Shetye, and T.~J. Overbye, ``{Grid
  Structural Characteristics as Validation Criteria for Synthetic Networks},''
  \emph{{IEEE Transactions on Power Systems}}, vol.~32, no.~4, pp. 3258--3265,
  July 2017.

\bibitem{TestCase}
S.~Maslennikov, B.~Wang, Q.~Zhang, F.~Ma, X.~Luo, K.~Sun, and E.~Litvinov, ``{A
  Test Cases Library for Methods Locating the Sources of Sustained
  Oscillations},'' in \emph{{IEEE PES General Meeting, Boston, MA}}, 2016.

\bibitem{Hankel}
S.~Zhang, Y.~Hao, M.~Wang, and J.~H. Chow, ``{Multi-Channel Hankel Matrix
  Completion through Nonconvex Optimization},'' \emph{{IEEE Journal of Selected
  Topics in Signal Processing}}, vol.~12, no.~4, pp. 617--632, 2018.

\end{thebibliography}
		
			\end{document}